\documentclass[journal,10pt]{IEEEtran}

\usepackage[T1]{fontenc}		
\usepackage[utf8]{inputenc}	

\usepackage{cite}
\usepackage{graphicx}
\usepackage{amsmath,amssymb}
\usepackage{amsfonts}
\usepackage{amsthm}

\usepackage{hyperref}
\usepackage{subcaption}
\usepackage{nicefrac}
\usepackage{placeins} 
\usepackage[ruled,vlined,linesnumbered]{algorithm2e}
\usepackage[dvipsnames]{xcolor}

\newtheorem{theorem}{Theorem}

\usepackage{amsmath,lipsum}
\usepackage{cuted}
\usepackage{url}

\begin{document}
\title{Effect of Strong Time-Varying Transmission Distance on LEO Satellite-Terrestrial Deliveries}
\author{Yuanyuan Ma,~Tiejun~Lv,~\IEEEmembership{Senior Member,~IEEE}, Tingting Li,~Gaofeng Pan,\\
\IEEEmembership{Senior Member,~IEEE},~Yunfei~Chen,~\IEEEmembership{Senior Member,~IEEE},~and~Mohamed-Slim Alouini,~\IEEEmembership{Fellow,~IEEE}
\thanks{Manuscript received November 8, 2021; revised March 31, 2022, and May 30, 2022; accepted June 8, 2022. This work was supported  by the National Natural Science Foundation of China under Grant 62171031. \emph{(Corresponding author: Tiejun Lv)}.}
\thanks{Y. Ma and T. Lv are with the School of Information and Communication Engineering, Beijing University of Posts and Telecommunications, Beijing 100876, China (e-mail: e-mail: \{mayuan, lvtiejun\}@bupt.edu.cn).}
\thanks{T. Li is with the School of Mathematics and Statistics, Southwest University, Chongqing 400715, China.}
\thanks{G. Pan is with the School of Cyberspace Science and Technology, Beijing Institute of Technology, Beijing 100081, China.}
\thanks{Yunfei Chen is with the School of Engineering, University of Warwick, Coventry CV4 7AL, U.K.}
\thanks{Mohamed-Slim Alouini is with the Computer, Electrical, and Mathematical Sciences and Engineering Division, King Abdullah University of Science and Technology (KAUST), Thuwal 23955-6900, Saudi Arabia.}
}
\maketitle
\begin{abstract}
In this paper, we investigate the effect of the strong time-varying transmission distance on the performance of the low-earth orbit (LEO) satellite-terrestrial transmission (STT) system.  We propose a new analytical framework  using finite-state Markov channel (FSMC) model and time discretization method. Moreover, to demonstrate the applications of the proposed framework, the performances of two adaptive transmissions, rate-adaptive transmission (RAT) and power-adaptive transmission (PAT) schemes,   are evaluated  for the cases when the transmit power or the transmission rate at the LEO satellite is fixed. Closed-form expressions for the throughput, energy efficiency (EE), and delay outage rate (DOR) of the considered systems are derived and verified, which are capable of addressing the capacity, energy efficiency, and outage rate performance of the considered LEO STT scenarios with the proposed analytical framework.

\end{abstract}
\begin{IEEEkeywords}
delay outage rate (DOR), energy efficiency (EE), finite-state Markov channel (FSMC), low-earth orbit (LEO) satellite, throughput, time variability.
\end{IEEEkeywords}

\section{Introduction}
\global\long\def\figurename{Fig.}

Satellite communication is becoming an important technology in the beyond five-generation (B5G) and six-generation (6G) systems to support the exponentially increasing data requirement and variety of users across the world, since it has been widely applied in mass broadcasting, navigation, and disaster relief operations for its capability of seamless connectivity and wide coverage  \cite{Ye1,6G1,6G2,Ye4,9760009}.

Among various research topics, satellite-terrestrial transmission (STT) system has attracted a significant amount of attention on its performance analysis, such as capacity \cite{multisat_multibeam,CDMA_cellular,Lin_BF_rate,An_add_re3}, energy efficiency (EE) \cite{EE_SIOT,EE_ref_Z_Lin2},  outage probability (OP) \cite{Guo_SOP,DVB-S2,Pan_3,An_add_re3}, and coverage probability (CP) \cite{Tian_Yu}. Most of these works have considered simple satellite-terrestrial systems, in which the transmission between satellites and ground users is direct without the assistance of ground relays. For instance, the authors of \cite{multisat_multibeam} analyzed the system capacity of the forward link of a code-division multiple access system.  The capacity of mobile satellite systems was shown in \cite{CDMA_cellular} with the assistance of adaptive power control for the attenuation caused by fading. The authors of \cite{EE_SIOT} studied the EE in satellite-based internet of things system with a proposed network coding hybrid automatic repeat request (HARQ) scheme. Ref. \cite{EE_ref_Z_Lin2} investigated the EE of the earth station  in the rate-splitting multiple access-based cognitive satellite-terrestrial networks in the presence of multiple eavesdroppers. Closed-form expressions for the average secrecy capacity and secure outage probability of the multiuser downlink wiretap satellite network were derived in \cite{Guo_SOP}. The OP was investigated in non-orthogonal multiple access (NOMA) based cooperative STT system \cite{An_add_re3,Add_re1_4,Add_re1_6}. The OP of a broadcasting satellite communication system and that of a cooperative satellite-aerial-terrestrial system were studied in \cite{DVB-S2,Pan_2}. The outage performance and diversity gain of three HARQ schemes were  respectively studied in \cite{Pan_3}.
\cite{Tian_Yu} derived the CP of a dual-hop cooperative satellite-unmanned aerial vehicle communication system.
The aforementioned works considered the impacts of the small-scale fading and/or the randomness of terrestrial position on satellite-terrestrial transmissions. Some literatures employed GEO satellite, in which the transmission distance can be viewed as fixed. However, \cite{An_add_re3,DVB-S2} employed LEO satellite, in which the impacts of time-varying transmission distance was ignored.

Low-earth orbit (LEO) satellites normally orbit hundreds of kilometers above the earth and LEO satellite communication systems have recently gained great research interest due to less power consumption, lower transmission delays, and higher data rates, which are expected to be incorporated in future wireless networks \cite{LEO_6,LEO_8,LEO,Add_re1_1,Add_re1_2}. Similar to traditional terrestrial wireless communication systems, they also suffer from two types of channel fading, namely, small-scale fading and large-scale fading. A shadowed Rice (SR) model \cite{Rice,satellite-link,Add_re1_3,Add_re1_5} is widely adopted for small-scale fading over LEO STT links, which describes the statistical distribution of the channel gain between the LEO satellite and terrestrial terminals. Furthermore, as pointed out in \cite{Yejia}, a notable characteristic of large-scale fading in LEO STT systems is time variability, due to the high mobility of the LEO satellites. Specifically, the strong time-varying transmission distance has a huge impact on the large-scale fading. However, till now, few works have analyzed the influence of strong time-varying transmission distance on the performance of the LEO STT system.


On the other hand, network resources can be adaptively allocated based on the varying channel conditions, for improved/optimal efficient resource configuration and management \cite{DBLP}. Ref. \cite{ACM_Sate_EHF} proposed a fully rate adaptive technique for use at extremely high-frequency bands which experience a high rate of change of rain fade levels. An energy-efficient adaptive transmission scheme for the integrated satellite-terrestrial network constraints has been proposed in \cite{EE_adaptive}. Ref. \cite{ACM_multibeam} proposed a cross-layer design with rate-adaptive transmission for the broadcast channel of an interactive multibeam broadband satellite system with a transparent architecture. Ref. \cite{Power_Control_Cognitive_Satellite} investigated the distributed power control problem with adaptive learning algorithm in downlink cognitive satellite-terrestrial networks. Moreover, the authors of \cite{adaptive} analyzed the channel capacity of a hybrid satellite-terrestrial network with different adaptive transmission schemes including adaptive transmit power and adaptive rate schemes.

However, the study of large-scale fading with LEO satellite  usually fails to  disregard the strong time variability \cite{An_add_re3,DVB-S2}, and the investigation of  the  small-scale fading  always ignore the practicability brought by quantification or  relationship between different states \cite{Add_re1_5}.
Motivated by these observations,  we propose a new analytical framework to study the influences of the strong time variability on the LEO STT link, coming from the large-scale fading and the small-scale fading  in time domain and  amplitude domain with time discretization method and the adopted FSMC model.  The throughput, EE, and delay outage rate (DOR) of the considered LEO STT system are analyzed, while considering two adaptive transmission schemes, rate-adaptive transmission (RAT) and power-adaptive transmission (PAT). The main contributions of this paper are summarized as follows:

\begin{itemize}
\item We propose a new analytical framework to accurately  analyze the performances of LEO STT systems, in which the large-scale fading (or the path-loss) shows strong time-varying properties. Specifically, the time discretization method is employed to model the effect of time-varying path-loss over STT links, and the FSMC method is adopted to study the effect of the small-scale fading over LEO STT links.

\item Closed-form expressions for the throughput, EE, and DOR of the considered system are derived and verified.

\item The performances of two adaptive transmission schemes, RAT and PAT schemes, in the LEO STT system are demonstrated and compared.
\end{itemize}

The remainder of this paper is organized as follows.  The considered LEO STT system and performance metrics are introduced in Section II and III, respectively. The performance analysis with RAT and PAT schemes are respectively analyzed in Sections IV and V. In Section VI, numerical results are presented to show the performances with two adaptive transmission schemes. Finally, Section VII concludes the paper.

\section{System model}
The LEO STT system considered in this work is presented in Fig. \ref{system}.  All of the  LEO satellites have the same orbit height denoted as $H$.
Although multi-beam satellite systems have high spectrum efficiency \cite{wc8869710}, the intractable inter-beam interference  will be induced. For simplification of the analysis, the single-beam satellite antennas are adopted in the considered STT system.
The radius of the earth is ${R}_\mathrm{e}$. The terrestrial terminal is denoted as $\mathrm{T}$, and the LEO satellite is denoted as $\mathrm{S}$.
Assume that the overlapping areas of different satellites sharing the same frequency resources are small enough to ignore inter-satellite interferences\footnote{If the overlapping areas are large so that the inter-satellite interference cannot be ignored, the interference mitigation technique in \cite{Interfe_Miti} for spectral coexistence  between LEO satellites can be employed to remove the interferences.}.

Let $x(t)$ be the information bit transmitted by $\mathrm{S}$ in time slot $t$. Consider a frequency nonselective fading channel with additive white Gaussian  noise $z\left(t\right)$. The received signal at $\mathrm{T}$ in time slot $t$ can be written as
\begin{equation}
y\left(t\right)=\sqrt{{P_{\mathrm{T}}}\left(t\right)/{d_{\mathrm{\mathrm{TS}}}^{\rho}(t)}}h_{\mathrm{TS}}\left(t\right)x\left(t\right)+z\left(t\right),
\end{equation}
where $d_{\mathrm{TS}}(t)$ is the distance between
$\mathrm{S}$ and $\mathrm{T}$, $\rho$ is the
path-loss factor, ${d_{\mathrm{\mathrm{TS}}}^{\rho}(t)}$ is the large-scale fading,
$P_{\mathrm{T}}(t)$ is the transmit power at ${\mathrm{S}}$, $h_{\mathrm{TS}}(t)$ is the small-scale fading channel gain between $\mathrm{T}$ and $\mathrm{S}$.

The received signal to noise ratio (SNR) at $\mathrm{T}$ is
\begin{equation}
\gamma(t)=\frac{P_{\mathrm{T}}(t)}{\sigma^{2}}\frac{|h_{\mathrm{TS}}(t)|^{2}}{d_{\mathrm{\mathrm{TS}}}^{\rho}(t)},
\end{equation}
where $\sigma^{2}$ is the average power of $z\left(t\right)$, the power gain $|h_{\mathrm{TS}}(t)|^{2}$ is from the small-scale fading and the pass loss $d_{\mathrm{\mathrm{TS}}}^{\rho}(t)$ is from large-scale fading/path-loss.

\begin{figure}
\begin{centering}
\includegraphics[width=9cm]{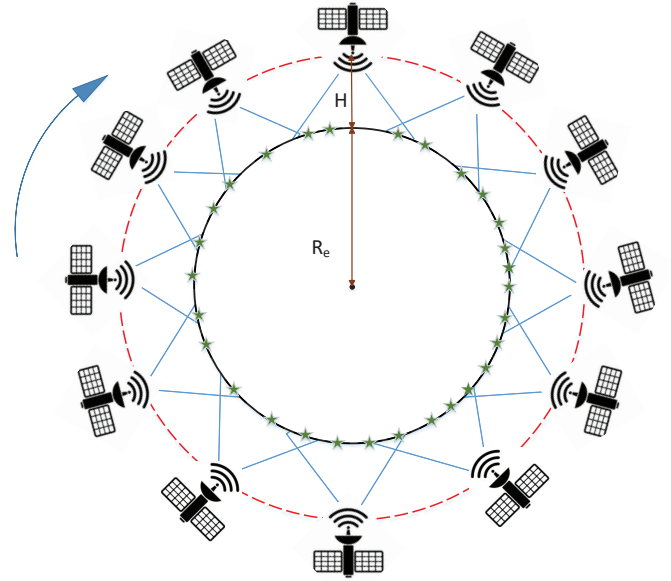}
\par\end{centering}
\caption{Adopted LEO STT system model: Several LEO satellites circling the earth make up an annular covering belt which afford uninterrupted service for the terrestrial terminal}
\label{system}
\end{figure}

\subsection{Channel Model}
 In this work, we consider the impacts of two types of channel fading in LEO STT system: the small-scale and large-scale fading. The SR model \cite{Rice,satellite-link} is adopted to describe the small-scale fading, which has been proved to be an accurate, practical and applicable tool to evaluate the performance of the satellite propagation environments in various frequency bands. Without loss of generality, the probability density function (PDF) of the power gain, $G\left(t\right)=|h_{\mathrm{TS}}\left(t\right)|^{2}$, for the LEO STT links in SR
  fading is given as  \cite{Rice}
\begin{align}\label{f_fading}
f_{\mathrm{G}}(y,t)=\alpha(t)\exp(-\beta(t) y)_{1}F_{1}\left(m(t);1;\delta(t) y\right),y\geq0 ,
\end{align}
where $\alpha(t)=\left(\frac{2b_{0}(t)m(t)}{2b(t)m(t)+\varOmega(t)}\right)^{m(t)}/(2b_{0}(t))$, $\beta(t)=\frac{1}{2b_{0}(t)},$ and $\delta(t)=\frac{\varOmega(t)}{2b_{0}(t)(2b_{0}(t)m(t)+\varOmega(t))}$, $\varOmega(t)$ and $2b_{0}(t)$ are are the average power of the line of sight (LOS) and multi-path components at time $t$, respectively, $m(t)$ is the fading severity parameter, and $_{1}F_{1}\left(\cdot;\cdot;\cdot\right)$ is the confluent hypergeometric function of the first kind \cite[Eq. (9.21)]{Gradshteyn}.

Expanding $_{1}F_{1}\left(\cdot;\cdot;\cdot\right)$ in  \eqref{f_fading} by using \cite[Eq. (9.210.1)]{Gradshteyn}, the PDF of $G\left(t\right)$ can be presented as

\begin{align}\label{pdf_G}
f_{\mathrm{G}}\left(x,t\right)=\alpha(t)\sum_{k=0}^{m(t)-1}\varsigma\left(\mathrm{\mathit{k,t}}\right)x^{k}\exp\left(-\left(\beta(t)-\delta(t)\right)x\right)
,
\end{align}
where $\varsigma\left(\mathrm{\mathit{k,t}}\right)=\frac{\left(-1\right)^{k}\left(1-m(t)\right)_{k}\delta(t)^{k}}{(k!)^{2}}$
and $\left(w\right)_{k}=w\left(w+1\right)\cdots\left(w+k-1\right)$
is the Pochhammer symbol \cite{Gradshteyn}.

By substituting \eqref{pdf_G} into $F_{\mathrm{G,t}}\left(x\right)=\int_{0}^{x}f_{\mathrm{G}}\left(y,t\right)dy$, the cumulative distribution function (CDF) of $G\left(t\right)$  is obtained as
\begin{align}\label{f_fading_inte}
F_{\mathrm{G}}\left(x,t\right)  =& 1-\alpha(t)\sum_{k=0}^{m(t)-1}\varsigma\left(\mathrm{\mathit{k,t}}\right)\sum_{p=0}^{k}\frac{k!}{p!}x^{p}
\notag\\
&\times
\frac{\exp(-\left(\beta(t)-\delta(t)\right)x)}{\left(\beta(t)-\delta(t)\right)^{k+1-p}}.
\end{align}

\subsection{Finite-State Markov Channel (FSMC) Model }
Here we adopt the FSMC model \cite{finite-state-Markov} to study the effect  of the small-scale fading over LEO STT links.
To the randomness of the small-scale fading of the LEO STT link, we partition the channel gain  between $\mathrm{T}$ and $\mathrm{S}$, $h_{\mathrm{TS}}$, into $K$ regions, $h_{k}\in[\mu_{k-1},\mu_{i})$, $i=1,2,\cdots,K$, where the region boundaries, $\mu_{i}$, denoting small fading amplitude partition thresholds, are set to the minimum channel gain value required to achieve a target bit error rate (BER) value \cite{Goldsmith2005Wireless}, with $\mu_{0}=0$ and $\mu_{K}=\infty$. The number of FSMC states is  $K$. The FSMC is said to be in state $S=k$ if the channel gain $h_{\mathrm{TS}}$ falls in the region $h_{k}\in[\mu_{k-1},\mu_{k})$.

It is assumed that at time index $l$, the state of the channel is   $s_{k}=k$, $k\in\left\{ 0,\cdots,K-1\right\} $. The channel state is the output of a Markov chain. The steady-state probability of each FSMC state is derived by integrating the PDF of $G\left(t\right)$, over the corresponding region as

\begin{align}\label{pi_kt_def}
\mathbf{\pi}_{k}(t) &= \int_{\mu_{k-1}^2}^{\mu_{k}^2}f_{{\mathrm{G}}}(x,t)dx\notag\\
& =\begin{cases}
F_{\mathrm{G}}\left({\mu_{k}^{2}},t\right)-F_{\mathrm{G}}\left({\mu_{k-1}^{2}},t\right), & k= 1,\cdots,K-1; \\
1-F_{\mathrm{G}}\left({\mu_{k-1}^{2}},t\right), & k=K
\end{cases},
\end{align}
where $f_{h_{\mathrm{TS}}}(x,t)$ is the the analytical PDF of the channel fading envelope. For stationary Markov models, the stationary state probability vector is denoted as $\boldsymbol{\pi}=[\pi_{1}(t),\cdots,\pi_{K}(t)]$.

\subsection{Time Discretization}
The strong time-varying transmission distance exhibits a nonnegligible effect on the path-loss because of the orbiting of the LEO satellite. To describe the time-varying property of the LEO STT channel, we need to partition the service time of $\mathrm{S}$ into time slots. Then, the influence of the strong time variability of the path-loss on the LEO STT link can be elaborated via employing the time discretization method. For clarity, we present the coverage area of one satellite in Fig. \ref{coverage}, in which the coverage area of $\mathrm{S}$ is a circle with a radius $\mathrm{R}$.

To facilitate the following analysis, we denote the sub-satellite point on earth as $\mathrm{O}$, the track of sub-satellite point as $\mathrm{TR_{-}O}$, the projection point of $\mathrm{T}$ on $\mathrm{TR_{-}O}$ as $\mathrm{P}$, the distance between $\mathrm{T}$ and $\mathrm{TR_{-}O}$ as $d_{\mathrm{P}}$, and the distance between $\mathrm{T}$ and $\mathrm{S}$ as $d_{\mathrm{TS}}$. We use $d_{\mathrm{P}}$ to describe the location of $\mathrm{T}$. Moreover, it is also assumed that $d_\mathrm{P}$ is fixed when $\mathrm{T}$ is in the coverage area of each LEO satellite, or $\mathrm{T}$ is fixed. The distance between $\mathrm{S}$  and $\mathrm{T}$ is calculated in Appendix A.

\begin{figure}
\begin{centering}
\includegraphics[width=7cm]{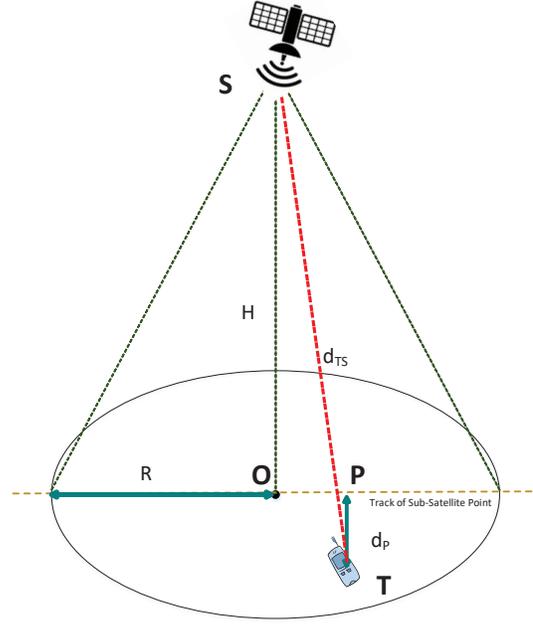}
\par\end{centering}
\caption{Coverage area of the LEO Satellite}
\label{coverage}
\end{figure}
When the terminal is in the service area of the LEO satellite, the channel gain and distance between $\mathrm{T}$ and $\mathrm{S}$ vary with the time. The  time when $\mathrm{T}$ is within the coverage area of each LEO satellite is denoted as $T_{\mathrm{\mathrm{s}}}$. We define the slot duration as $T_{\mathrm{slot}}$. Then, the total duration can be divided into $N$ slots as $N=\frac{T_{\mathrm{s}}}{T_{\mathrm{slot}}}$, where $N$ is an integer. In each slot, we assume the channel gain $|h_{\mathrm{TS}}|$ belongs to one $h_{k}\overset{\triangle}\in [\mu_{k-1},\mu_{k})$, $k\in\left\{ 1,\cdots,K\right\} $. Similarly, $d_{\mathrm{TS}}[n]$ is denoted as the distance between $\mathrm{S}$ and $\mathrm{T}$ in the $n$-th time slot and $d_{\mathrm{TS}}[n]\in \left[ \mathrm{min}\left\{ d_{\mathrm{TS}}\left(t\right)\right\},  \mathrm{max}\left\{ d_{\mathrm{TS}}\left(t\right)\right\} \right]$, in which ${\left(n-1\right)T_{\mathrm{slot}}}\leq t\leq{nT_{\mathrm{slot}}}$.

Considering the time-varying property of the LEO STT channel and   practicability of the considered LEO STT system simultaneously,
 we know the large-scale fading and the small-scale fading of the LEO STT system  all vary in time and  amplitude domains, and can be presented with  time discretization method and the adopted FSMC model. Next, the throughput, EE and DOR of the considered LEO STT system  will be analyzed with FSMC model and  time discretization method.

\section{Throughput, EE and DOR}

Using the channel gain discretization in the FSMC model and time discretization to describe the time variability of the path-loss,  the throughput and EE of the considered system are
\begin{equation}\label{rate_bar}
\bar{R}=\frac{1}{N}\sum_{n=1}^{N}\sum_{k=1}^{K}\pi_{k,n}R_{k,n},
\end{equation}
and
\begin{align}\label{EE_bar}
\eta_{\rm EE} =\frac{\sum_{n=1}^{N}\sum_{k=1}^{K}\pi_{k,n}R_{k,n}}{\sum_{n=1}^{N}\sum_{k=1}^{K}\pi_{k,n}P_{k,n}},
\end{align}
respectively,
where
$\pi_{k,n}$, $R_{k,n}$ and $P_{k,n}$ are the stationary state probability,  the data rate,  and the transmit power of ${\mathrm{S}}$, respectively,  when ${\mathrm{T}}$ is in the $n$-th time slot and the channel gain is $h_{k}$, in which
$\pi_{k,n}=\frac{1}{T_{\mathrm{slot}}}\int_{\left(n-1\right)T_{\mathrm{slot}}}^{nT_{\mathrm{slot}}}\pi_{k}(t)$.

DOR is defined as the probability that the time duration required to successfully deliver a certain amount of data  D over a wireless channel is greater than a threshold, denoted by $T_{\mathrm{th}}$. Given that $T_{\mathrm{th}}$ can be related to the latency requirement of the data, the DOR metric simultaneously characterizes the reliability and latency performances of transmission. DOR can be expressed as
\begin{align}\label{DOR}
{\mathrm{DOR}}=\Pr\left\{ \mathrm{DT}>T_{\mathrm{th}}\right\}.
\end{align}

The delivery time of a small data packet $D$ arriving at the time $t$  is denoted by ${\mathrm{DT}}(t)$.
According to the Markov model for the fading channel where channel gain can only move to the neighboring intervals, $h_{\mathrm{TS}}\left(t\right)$
enters region $h_{2}$ after the waiting period and rate $R_{2,m_t}$ is used to complete the data transmission, where
\begin{align}\label{m_equation}
m_t=\left\lceil \frac{T_{\mathrm{W}}+t}{T_{\mathrm{slot}}}\right\rceil \,\mathrm{MOD}\,N
\end{align}
indicating the time slot index after the waiting period and $m_t=1,2,\cdots,N$.
In other words, $m_t$ indicates the  large-scale fading with time slot index.
Hence, we have

\begin{align}\label{DT_DEFINE}
\mathrm{DT}(t)=\begin{cases}
\frac{D}{R_{k,m_t}}, & k=2,\cdots,K ;\\
T_{\mathrm{W}}+\frac{D}{R_{k,m_t}},& k=1
\end{cases},
\end{align}
in which $T_{\mathrm{W}}$ is the waiting period when the data packet arrives at time $t$. We adopt the Markov channel model \cite{Queueing,finite-state-Markov} and assume that, when a small data packet arrives at time $t$ , the channel gain stays in $h_{1}$ for an exponentially distributed period of  time with PDF

\begin{align}\label{f_wait_time}
f_{T_{\mathrm{W}}}(t)=\frac{1}{\lambda}\exp\left(-\frac{t}{\lambda}\right),
\end{align}
in which  $\lambda$ is  the average of $T_{\mathrm{W}}$. In other words, $\lambda$ is the average time period over which $h(t)$ stays below a given threshold $\mu_1$ per unit time.

Here we introduce second-order statistics, the level crossing rate (LCR), which is the rate at which the envelope crosses a certain threshold, and the envelope average fade duration (AFD), which is the length of the time that the envelope stays below a given threshold. From  \eqref{f_wait_time},  $\lambda$  is AFD and it can be expressed as

\begin{align}\label{ave_wait_time_define}
\lambda=\frac{1}{N_R({\mu_1})}\int_{0}^{{\mu_1}^2}f_{\mathrm{G}}(y)dy=\frac{\pi_{1}}{N_R({\mu_1})},
\end{align}
in which $N_R({\mu_1})$ is LCR, the rate at which the envelope of $h(t)$ crosses a certain threshold ${\mu_1}$.

In this paper, we consider the small-scale fading as the slowly varying LOS and non isotropic scattering correlation model. We denote $f_{\mathrm{max}}^{\mathrm{scatter}}$ and $f_{\mathrm{max}}^{\mathrm{LOS}}$  as the maximum Doppler frequency of the scattering component and the maximum Doppler frequency of the LOS component. Empirical observations have shown that the rate of change of the LOS component (several Hertz), is significantly less than that of the scattering component (several hundred Hertz)\cite{kappa_value}, which means $\frac {f_{\mathrm{max}}^{\mathrm{LOS}}}{f_{\mathrm{max}}^{\mathrm{scatter}}}{\ll}1$. Thus, \cite{Rice} gives LCR $N_R({r_{th}})$ as

\begin{align}
N_R({r_{th}})=\frac{1}{\sqrt{2\pi}\Gamma\left(m\right)}\left(\frac{2b_{0}m}{2b_{0}m+\Omega}\right)^{m}\sqrt{\frac{b_{0}b_{2}-b_{1}^{2}}{b_{0}}}\frac{r_{th}}{b_{0}}\notag\\
\times\exp\left(-\frac{\mu_{k}^{2}}{b_{0}}\right)\sum_{n=0}^{\infty}\frac{\left(\frac{1}{2}\right)_{n}\left(-1\right)^{n}}{n!}\left[\xi_{n}\left({r_{th}}\right)+\xi_{n+1}\left({r_{th}}\right)\right],
\end{align}
where $\left(x\right)_{n}=x\left(x+1\right)\cdots\left(x+n-1\right)$,
$\left(x\right)_{0}=1$, $b_{1}=b_{0}2\pi f_{\mathrm{max}}^{\mathrm{scatter}}\cos(\bar{\phi})I_{1}(\kappa)/I_{0}(\kappa)$, $b_{2}=b_{0}2\pi^{2}f_{\mathrm{max}}^{\mathrm{scatter}^{2}}[I_{0}(\kappa)+\cos(\bar{\phi})I_{2}(\kappa)]/I_{0}(\kappa)$,
, $\bar{\phi}\in\left[-\pi\right. \left.,\pi\right)$ is the mean direction of the angle of arrival (AOA) in the horizontal plane, and $\kappa$ is the width control parameter of the AOA \cite{AOA}, $I_{0}(\cdot)$ is the modified Bessel function of zeroth order and

\begin{align}
\xi_{n}\left({r_{th}}\right)=&\frac{\Gamma\left(n+m\right)}{2^{n}n!}\left[\frac{b_{1}^{2}}{b_{0}\left(b_{0}b_{2}-b_{1}^{2}\right)}\right]^{2}\left(\frac{2b_{0}\Omega}{2b_{0}m+\Omega}\right)^{n} \notag\\
& \times {}_{1}F_{1}\left(n+m,n+1,\frac{\Omega{r_{th}}^{2}}{2b_{0}\left(2b_{0}m+\Omega\right)}\right),
\end{align}
in which  $_{1}F_{1}\left(\cdot;\cdot;\cdot\right)$ is the confluent hypergeometric function of the first kind \cite[Eq. (9.21)]{Gradshteyn}.


To comprehensively reflect the time-varying property of the LEO STT channel, in the following two sections we will demonstrate the performances of two adaptive transmission schemes, RAT and PAT  for the cases when the transmit power or the transmission rate at the LEO satellite is fixed.  Closed-form expressions for the throughput, EE, and DOR of the considered LEO STT system will be derived in each adaptive transmission scheme.
Besides, we assumed that
Doppler frequency shift, caused by the mobility of the satellite, can be estimated perfectly and mitigated by the mature
precompensation method \cite{Add_re1_5,Doppler_simu_estimate}.
Moreover, we assume that the channel is perfectly estimated with pilot signals \cite{sat_perfect} at the terrestrial terminal and known to the terrestrial terminal, through a noiseless feedback channel. It means that the LEO satellite has full channel state information (CSI) to implement adaptive transmission\footnote{
Due to the propagation delay between the LEO satellite and the terrestrial terminal, we can obtain the outdated CSI in practice \cite{outdated_CCI_OutCSI}.
It is obvious that the propagation delay in LEO  STT  is relatively short compared to the slot duration used in time discretization.
According to the FSMC model and the lower and upper bounds of the received SNR based on the region boundaries of CSI, the outdated CSI can be regarded as instantaneous CSI with acceptable error.}.  The impact of imperfect CSI is out of the scope of this work and the reader interested in this issue can refer to Refs. \cite{imp_csi1,imp_csi2} for further information.

\section{RATE-ADAPTIVE TRANSMISSION (RAT)}
We now consider a practical discrete-rate adaptive transmission with AMC \cite{Goldsmith2005Wireless} where the transmit power at $\mathrm{S}$ is fixed. In the following analysis, closed-form expressions for the throughput, EE, and DOR of the considered LEO STT system are derived.

\subsection{Throughput}
To ensure high transmission reliability, we assume that no transmission occurs when $\gamma<\gamma_{\mathrm{min}}$. The delivery time of a small data packet and the data rate with AMC depends on the received
SNR.

According to the aforementioned model, the  range of the received SNR is divided into $K$ regions in each time slot similar to the channel gain. The region boundaries are set to the minimum SNR value required by the selected modulation and coding scheme to achieve a target BER value $\mathrm{BER_{tar}}$ \cite{Goldsmith2005Wireless}. To achieve the throughput, we set
\begin{equation}\label{mu_1}
\mu_{1}=\sigma\sqrt{\frac{\gamma_{\mathrm{min}}d_{\mathrm{max}}^{\rho}}{P_{\mathrm{T}}}}.
\end{equation}

When the channel gain falls in the region $h_{k}$ and $\mathrm{T}$ is in the $n$-th time slot, the lower and upper bounds of received SNR are

\begin{equation}\label{gamma_k_n_L}
\gamma_{k,n}^{\mathrm L}=\begin{cases}
0, & k=1;\\
\frac{P_{\mathrm{T}}}{\sigma^{2}}\frac{\mu_{k-1}^{2}}{d_{\mathrm{TS}}^{\rho}\left[n\right]}, & k=2,\cdots,K
\end{cases}
,\end{equation}
and
\begin{equation}\label{gamma_k_n_U}
\gamma_{k,n}^{\mathrm U}=\begin{cases}
0, & k=1;\\
\frac{P_{\mathrm{T}}}{\sigma^{2}}\frac{\mu_{k}^{2}}{d_{\mathrm{TS}}^{\rho}\left[n-1\right]}, & k=2,\cdots,K
\end{cases},
\end{equation}
respectively.

The lower and upper bounds of data rate are
 \begin{equation}
R_{k,n}^{\mathrm{RAT,L}}=B\log_{2}\left(1+\gamma_{k,n}^{L}\right)
,\end{equation}
and
 \begin{equation}
R_{k,n}^{\mathrm{RAT,U}}=B\log_{2}\left(1+\gamma_{k,n}^{U}\right),
\end{equation}
respectively, in which $n=1,\cdots,N$.
 Based on \eqref{pi_kt_def} and time discretization method, the probability of using $R_{k,n}^{\mathrm{RAT}}$ is given by
\begin{align}
\pi_{k,n}&=\mathbf{Pr}\left\{ h_{\mathrm{TS}}\in h_{k}\right\}
 \notag\\&
=F_{\mathrm{G}}\left(\mu_{k}^{2},t\right)-F_{\mathrm{G}}\left(\mu_{k-1}^{2},t\right),
\end{align}
where $F_{\mathrm{G}}\left(\mu_{k},t\right)$ is expressed as \eqref{f_fading} with non-integer $m(t)$ and \eqref{f_fading_inte} with integer $m(t)$.

Based on \eqref{rate_bar}, the lower and upper bounds of the average throughput  with RAT scheme are
\begin{equation}
\bar{R}_{1}^{\mathrm L}=\frac{1}{N}\sum_{n=1}^{N}\sum_{k=1}^{K}\pi_{k,n}R_{k,n}^{\mathrm{RAT,L}}
,\end{equation}
and
\begin{equation}
\bar{R}_{1}^{\mathrm U}=\frac{1}{N}\sum_{n=1}^{N}\sum_{k=1}^{K}\pi_{k,n}R_{k,n}^{\mathrm{RAT,U}},
\end{equation}
respectively.

\subsection{EE}
First, we will analyze the average power consumption $\bar{P}_{1}$ with RAT. As no transmission occurs when $\gamma<\gamma_{\mathrm{min}}$ to ensure high transmission reliability, the transmit power at the LEO satellite is zero when  $\gamma<\gamma_{\mathrm{min}}$ and is $P_{\mathrm{T}}$ when  $\gamma>\gamma_{\mathrm{min}}$. Let $P_{k,n}^\mathrm{RAT}$ be the transmit power of ${\mathrm{S}}$ when ${\mathrm{T}}$ in the $n$-th time slot and the channel gain stays as $h_{k}$, $k\in\left\{ 1,\cdots,K\right\} $.

\begin{equation}\label{P_kn_RAT}
P_{k,n}^{\mathrm{RAT}}=\begin{cases}
0, & \gamma<\gamma_{\mathrm{min}};\\
P_\mathrm{T}, & \mathrm{else}
\end{cases}.
\end{equation}

Combining \eqref{mu_1} and \eqref{P_kn_RAT}, $P_{k,n}^{\mathrm{RAT}}$ can be simplified as
\begin{align}
P_{k,n}^{\mathrm{RAT}}=\begin{cases}
0, & k=1;\\
P_\mathrm{T}, & k=2,\cdots,K
\end{cases}.
\end{align}

The average power consumption with RAT scheme is
\begin{align}
\bar{P}_{1} =\frac{1}{N}\sum_{n=1}^{N}\sum_{k=1}^{K}\pi_{k,n}P_{k,n}^{{\mathrm{RAT}}}\label{ave_P_RAT}.
\end{align}

Based on \eqref{EE_bar}, the lower and upper bounds of EE with RAT scheme are
\begin{align}
\eta_{{\rm EE}_{\mathrm{RAT}}}^{\mathrm L}&  =\frac{\sum_{n=1}^{N}\sum_{k=1}^{K}\pi_{k,n}R_{k,n}^{{\mathrm{RAT}},L}}{\sum_{n=1}^{N}\sum_{k=1}^{K}\pi_{k,n}P_{k,n}^{{\mathrm{RAT}}}}\label{EE_RAT}
,\end{align}
and
\begin{align}
\eta_{{\rm EE}_{\mathrm{RAT}}}^{\mathrm U} &  =\frac{\sum_{n=1}^{N}\sum_{k=1}^{K}\pi_{k,n}R_{k,n}^{{\mathrm{RAT}},U}}{\sum_{n=1}^{N}\sum_{k=1}^{K}\pi_{k,n}P_{k,n}^{{\mathrm{RAT}}}}\label{EE_RAT},
\end{align}
respectively.

\subsection{DOR}
The delivery time of a small data packet $D$ arriving at the time $t$  is denoted by ${\mathrm{DT}}_{\mathrm{R}}(t)$.
Based on \eqref{DT_DEFINE}, the probability mass function (PMF) of $\Pr\left\{ \mathrm{DT_{\mathrm{R}}}(t)\right\}$ can be written as

\begin{equation}
\Pr\left\{ \mathrm{DT_{\mathrm{R}}}(t)=\frac{D}{R_{k,m_t}^{\mathrm{RAT}}}\right\} =\pi_{k,m_t}, k=2,\cdots,K
,
\end{equation}
and

\begin{equation}
\Pr\left\{ \mathrm{DT_{\mathrm{R}}}(t)=T_{\mathrm{W}}+\frac{D}{R_{2,m_t}^{\mathrm{RAT}}}\right\} =\pi_{1,m_t}.
\end{equation}

Therefore, the CDF of $\mathrm{DT_{\mathrm{R}}}(t)$, denoted by $F_{\mathrm{DT_{\mathrm{R}}}(t)}(s),$ can be derived as

\begin{align}
F_{\mathrm{DT_{\mathrm{R}}}(t)}(s)=&\pi_{1,m_t}\Pr\left\{ T_{\mathrm{W}}+\frac{D}{R_{2,m_t}^{\mathrm{RAT}}}<s\right\} \notag\\
&+\sum_{k=2}^{K}\pi_{k,m_t}\mathcal{U}\left(s-\frac{D}{R_{k,m_t}^{\mathrm{RAT}}}\right),
\end{align}
where $\mathcal{U}(\cdot)$ denotes the unit step function. When the data packet arrives at the time $t$, the DOR with RAT scheme is denoted by $\mathrm{DOR}_{\mathrm{R}}(t)$,
which can be calculated as

\begin{align}\label{dor_def}
&\mathrm{DOR_{\mathrm{R}}}(t)=\Pr\left\{ \mathrm{DT_{\mathrm{R}}}>T_{\mathrm{th}}\right\}\notag\\
&=1-F_{\mathrm{DT_{\mathrm{R}}}(t)}(T_{\mathrm{th}})\notag\\
&=1-\pi_{1,m_t}\left[1-\exp\left(-\frac{T_{\mathrm{th}}-\frac{D}{R_{2,m_t}^{\mathrm{RAT}}}}{\lambda}\right)\right] \notag\\
&~~~~\times \mathcal{U}\left(T_{\mathrm{th}}-\frac{D}{R_{2,m_t}^{\mathrm{RAT}}}\right)-\sum_{k=2}^{K}\pi_{k,m_t}\mathcal{U}\left(T_{\mathrm{th}}-\frac{D}{R_{k,m_t}^{\mathrm{RAT}}}\right),
\end{align}
where $m_t$ is described in \eqref{m_equation} and
\begin{align}\label{lamda_RAT}
\lambda=\frac{F_{\mathrm{G}}\left(\sigma^2{\frac{\gamma_{\mathrm{min}}d_{\mathrm{max}}^{\rho}}{P_{\mathrm{T}}}}\right)}{N_R\left({\sigma\sqrt{\frac{\gamma_{\mathrm{min}}d_{\mathrm{max}}^{\rho}}{P_{\mathrm{T}}}}}\right)}.
\end{align}

\begin{theorem}\label{DOR1_theorem}
A closed-form expression for the average DOR with RAT scheme $\mathrm{DOR_{\mathrm{R}}}$ of the LEO satellite system is given by \eqref{DOR1_expression}, shown on the top of next page.
\begin{figure*}
  \centering
\begin{align}\label{DOR1_expression}
&\mathrm{DOR_{\mathrm{\mathrm{R}}}}=1-\frac{1}{N}\sum_{n=1}^{N}\sum_{k=2}^{M}\pi_{k,n}U\left(T_{\mathrm{th}}-\frac{D}{B\log_{2}\left(1+\frac{P_{\mathrm{T}}}{\sigma^{2}}\frac{\mu_{k-1}^{2}}{d_{\mathrm{TS}}^{\rho}[n]}\right)}\right)\notag\\
&-\frac{1}{N}\sum_{m_t=1}^{N}\pi_{1,m_t}\left[1-\exp\left(-\frac{1}{\lambda}\left(T_{\mathrm{th}}-\frac{D}{B\log_{2}\left(1+\frac{\gamma_{\mathrm{min}}d_{\mathrm{max}}^{\rho}}{d_{\mathrm{TS}}^{\rho}[m_t]}\right)}\right)\right)\right]\mathcal{U}\left(T_{\mathrm{th}}-\frac{D}{B\log_{2}\left(1+\frac{\gamma_{\mathrm{min}}d_{\mathrm{max}}^{\rho}}{d_{\mathrm{TS}}^{\rho}[m_t]}\right)}\right)
\end{align}
\rule{18cm}{0.01cm}
\end{figure*}
\end{theorem}

\begin{proof}
See Appendix B.
\end{proof}

\section{POWER-ADAPTIVE TRANSMISSION (PAT)}
The PAT scheme,  one has $P_\mathrm{T}<P_\mathrm{max}$, where $P_\mathrm{max}$ is the maximum transmit power at the LEO satellite, the data rate is fixed and the transmit power at $\mathrm{S}$ varies. If the fixed data rate is $R_\mathrm{fix}$, when the channel gain is so small that the LEO STT link between  $\mathrm{T}$ and  $\mathrm{S}$ cannot afford reliable transmission, the reliable date rate is zero. In the following analysis, the throughput, EE, and DOR of PAT will be derived.

\subsection{Throughput and EE}
Similar to the RAT scheme, the  total service duration of $\mathrm{S}$ is divided into $N$ slots. In the $n$-th time slot, the distance between $\mathrm{T}$ and $\mathrm{S}$ is defined  $d_{\mathrm{TS}}[n]=\mathrm{max}\left\{ d_{\mathrm{TS}}\left(t\right)\right\} $,
in which $\frac{\left(n-1\right)T_{\mathrm{slot}}}{N}\leq t\leq \frac{nT_{\mathrm{slot}}}{N}$.
$P_{k,n}^\mathrm{PAT}$ is denoted as the transmit power at ${\mathrm{S}}$ when ${\mathrm{T}}$ in the  $n$-th time slot and the channel gain stays into $h_{k}$, $k\in\left\{ 1,\cdots,K\right\} $.

\begin{equation}
P_{k,n}^\mathrm{PAT}=\begin{cases}
0, & P_\mathrm{instan}>P_\mathrm{max};\\
P_\mathrm{instan}, & \mathrm{else}
\end{cases},
\end{equation}
in which
\begin{equation}
P_\mathrm{instan}=\frac{\sigma^{2}{d_{\mathrm{TS}}^{\rho}[n]}}{\mu_{k-1}^{2}}\left({2^{\frac{ R_\mathrm{fix}}{B}}-1}\right).
\end{equation}

When $P_{\mathrm{instan}}>P_{\mathrm{max}}$, $u_{k-1}<\sigma\sqrt{\frac{\left(2^{\frac{R_{\mathrm{fix}}}{B}}-1\right)d_{\mathrm{TS}}^{\rho}\left[n\right]}{P_{\mathrm{max}}}}$.
To get a simpler partition of channel gain, we set

\begin{align}
u_{1}=\sigma\sqrt{\frac{\left(2^{\frac{R_{\mathrm{fix}}}{B}}-1\right)d_{\mathrm{max}}^{\rho}}{P_{\mathrm{max}}}}.
\end{align}

When $\mathrm{T}$ is in the $n$-th time slot and the channel gain stays into $h_{k}$, the transmit power at $\mathrm{S}$ and the data rate used are

\begin{equation}\label{P_kn_up_low}
P_{k,n}^{\mathrm{PAT}}=\begin{cases}
0, & k=1;\\
\frac{\sigma^{2}d_{\mathrm{TS}}^{\rho}\left[n\right]}{u_{k-1}^{2}}\left(2^{\frac{R_{\mathrm{fix}}}{B}}-1\right), & \mathrm{upper},k=2,\cdots,K; \\
\frac{\sigma^{2}d_{\mathrm{TS}}^{\rho}\left[n\right]}{u_{k}^{2}}\left(2^{\frac{R_{\mathrm{fix}}}{B}}-1\right), & \mathrm{lower}, k=2,\cdots,K
\end{cases}
,\end{equation}
and
\begin{align}
R_{k,n}^{\mathrm{PAT}}=\begin{cases}
0, & k=1;\\
R_{\mathrm{fix}}, & k=2,\cdots,K
\end{cases}.
\end{align}

The average power of ${\mathrm{S}}$ under PAT scheme is
\begin{align}
\bar{P}_{2} =\frac{1}{N}\sum_{n=1}^{N}\sum_{k=1}^{K}\pi_{k,n}P_{k,n}^\mathrm{PAT}\label{ave_P_ada}.
\end{align}

Based on \eqref{rate_bar},  the average throughput with certain requirements of QoS $P_{\mathrm{max}}$ is
\begin{align}
\bar{R}_{2} =\frac{1}{N}\sum_{n=1}^{N}\sum_{k=2}^{K}\pi_{k,n}R_{k,n}^\mathrm{PAT}.\label{ave_R_ada}
\end{align}

Substituting (\ref{ave_P_ada}) and (\ref{ave_R_ada}) into (\ref{EE_bar}), EE can be derived as
\begin{align}\label{EE_PAT}
\eta_{{\rm EE}_{\mathrm{PAT}}}&  =\frac{\sum_{n=1}^{N}\sum_{k=2}^{K}\pi_{k,n}R_{k,n}^\mathrm{PAT}}{\sum_{n=1}^{N}\sum_{k=1}^{K}\pi_{k,n}P_{k,n}^\mathrm{PAT}},
\end{align}
which indicates the lower and upper bound with upper and lower bound of $P_{k,n}^{\mathrm{PAT}}$ in \eqref{P_kn_up_low}.
\subsection{DOR}
We assume that when a  small data packet under PAT scheme arrives at the time $t$, the channel gain stays as $h_{1}$ for an exponentially distributed period with the same PDF as the RAT scheme given by \eqref{f_wait_time}. The waiting period is denoted as $T_{\mathrm{W}}$. The delivery time of a small data packet arriving at the time $t$  with PAT scheme is denoted by ${\mathrm{DT}}_{\mathrm{P}}(t)$. Because the considered channel
gain can only move to the neighboring intervals, $h_{\mathrm{TS}}\left(t\right)$ enters region $h_{2}$ after the waiting period and rate $R_\mathrm{fix}$ is used to complete the data transmission. The delivery time of a small data packet ${\mathrm{DT}}_{\mathrm{P}}(t)$ is $T_{\mathrm{W}}+\frac{D}{R_{\mathrm{fix}}}$ when the  channel gain stays in $h_{1}$, and $\frac{D}{R_{\mathrm{fix}}}$ if the  channel gain stays in other regions.

\begin{equation}
\mathrm{DT_{\mathrm{P}}}(t)=\begin{cases}
\frac{D}{R_{\mathrm{fix}}}, & k=1;\\
T_{\mathrm{W}}+\frac{D}{R_{\mathrm{fix}}}, & k=2,\cdots,K
\end{cases},
\end{equation}
whose PMF can be written as

\begin{equation}
\Pr\left\{ \mathrm{DT_{\mathrm{P}}}(t)=\frac{D}{R_{\mathrm{fix}}}\right\} =1-\pi_{1}(t)
,\end{equation}
and

\begin{equation}
\Pr\left\{ \mathrm{DT_{\mathrm{P}}}(t)=T_{\mathrm{W}}+\frac{D}{R_{\mathrm{fix}}}\right\} =\pi_{1}(t).
\end{equation}

Therefore, the CDF of $\mathrm{DT_{\mathrm{P}}}$, denoted by $F_{\mathrm{DT_{\mathrm{P}}}(t)}(s),$ can be derived as
\begin{align}
F_{\mathrm{DT_{\mathrm{P}}}(t)}(s)=&\frac{1}{N}\sum_{n=1}^{N}\left[\pi_{1,n}\Pr\left\{ T_{\mathrm{W}}+\frac{D}{R_{\mathrm{fix}}}<s\right\}\right.
\notag\\
&~~~~~~~~~+\left.\sum_{k=2}^{K}\pi_{k,n}\mathcal{U}\left(s-\frac{D}{R_{\mathrm{fix}}}\right)\right],
\end{align}
where $\mathcal{U}(\cdot)$ denotes the unit step function.

Finally, the DOR with PAT scheme when the data arrives in the $n$-th time slot is denoted by $\mathrm{DOR}_{\mathrm{P}}\left[n\right]$, which can be calculated as

\begin{align}
&\mathrm{DOR_{\mathrm{P}}}(t)=\Pr\left\{ \mathrm{DT_{\mathrm{P}}}>T_{\mathrm{th}}\right\}\notag\\
&=1-\frac{1}{N}\sum_{n=1}^{N}\left\{\pi_{1,n}\left[1-\exp\left(-\frac{T_{\mathrm{th}}-\frac{D}{R_{\mathrm{fix}}}}{\lambda}\right)\right]\right.\notag\\
&~~~~\left.\times \mathcal{U}\left(T_{\mathrm{th}}-\frac{D}{R_{\mathrm{fix}}}\right)-\sum_{k=2}^{M}\pi_{k,n}\mathcal{U}\left(T_{\mathrm{th}}-\frac{D}{R_{\mathrm{fix}}}\right)\right\},
\end{align}
where
\begin{align}
\lambda=\frac{F_{\mathrm{G}}\left({\sigma^2{\frac{\left(2^{\frac{R_{\mathrm{fix}}}{B}}-1\right)d_{\mathrm{max}}^{\rho}}{P_{\mathrm{max}}}}}\right)}{N_R\left(\sigma\sqrt{\frac{\left(2^{\frac{R_{\mathrm{fix}}}{B}}-1\right)d_{\mathrm{max}}^{\rho}}{P_{\mathrm{max}}}}\right)}.
\end{align}

Then, the DOR with PAT scheme when the terminal is in the service area of the LEO satellite is shown on the top of next page in \eqref{DOR2}.
\begin{figure*}
  \centering
\begin{align}\label{DOR2}
\mathrm{DOR_{\mathrm{P}}}&=\frac{1}{T_{\mathrm{s}}}\int_{0}^{T_{\mathrm{s}}}\mathrm{DOR_{\mathrm{P}}}\left(t\right) dt \notag\\
&=1-\frac{1}{N}\sum_{n=1}^{N}\left\{\pi_{1,n}\left[1-\exp\left(-\frac{T_{th}-\frac{D}{R_{\mathrm{fix}}}}{\lambda}\right)\right]
+\sum_{k=2}^{M}\pi_{k,n}\right\}\mathcal{U}\left(T_{th}-\frac{D}{R_{\mathrm{fix}}}\right)\notag\\
&=\begin{cases}
1, & T_{\mathrm{th}}<\frac{D}{R_{\mathrm{fix}}};\\
\frac{1}{N}\sum_{n=1}^{N}\pi_{1,n}\exp\left(-\frac{T_{\mathrm{th}}-\frac{D}{R_{\mathrm{fix}}}}{\lambda}\right), & T_{\mathrm{th}}\geq \frac{D}{R_{\mathrm{fix}}}.
\end{cases}
\end{align}
\rule{18cm}{0.01cm}
\end{figure*}

To get a clear understanding of the aforementioned metrics of the considered LEO STT system, we put the  definitions of main variables related to throughput, EE, DOR  in Table I.

\begin{table}
\centering{}
\caption{
NOTATION SUMMARY}
\begin{tabular}{|c|c|}
\hline
\textbf{Symbol} & \textbf{Definition}\tabularnewline
\hline
$\bar{R}$ & throughput\tabularnewline
\hline
$\eta_{{\rm EE}}$ & energy efficiency\tabularnewline
\hline
DOR &  delay outage rate\tabularnewline
\hline
$\bar{R}_{1}^{\mathrm{L}}$ & lower bound of   throughput with RAT scheme \tabularnewline
\hline
$\bar{R}_{1}^{\mathrm{U}}$ &  upper bound of   throughput with RAT scheme \tabularnewline
\hline
$\eta_{{\rm EE}_{\mathrm{RAT}}}^{\mathrm{L}}$ & lower bound of EE with RAT scheme\tabularnewline
\hline
$\eta_{{\rm EE}_{\mathrm{RAT}}}^{\mathrm{U}}$ & upper bound of EE with RAT scheme\tabularnewline
\hline
$\mathrm{DOR_{\mathrm{\mathrm{R}}}}$ &  DOR with RAT scheme\tabularnewline
\hline
$\bar{R}_{2}$ &  throughput with PAT scheme\tabularnewline
\hline
$\eta_{{\rm EE}_{\mathrm{PAT}}}$ & EE with PAT scheme\tabularnewline
\hline
$\mathrm{DOR_{\mathrm{P}}}$ &  DOR with PAT scheme\tabularnewline
\hline
\end{tabular}
\end{table}

\section{Numerical Results And Disscussion}
\begin{table}
\centering{}
\caption{Values of the Simulation Parameters}
\begin{tabular}{|c|c|}
\hline
Parameter  & Value\tabularnewline
\hline
$m$ & $10.1$\tabularnewline
\hline
$\varOmega$ & $0.825$\tabularnewline
\hline
$b$  & $0.126$\tabularnewline
\hline
$B_\mathrm{W}$  & $60\,\mathrm{MHz}$\tabularnewline
\hline
$\sigma^{2}$  & $-66\,\mathrm{dBm}$\tabularnewline
\hline
$\gamma_{\mathrm{min}}$  & $0\,\mathrm{dB}$\tabularnewline
\hline
$\bar{\phi}$  & $1.55$\tabularnewline
\hline
$\kappa$  & $24.2$\tabularnewline
\hline
$T_{\mathrm{th}}$  & $0-1\,\mathrm{ms}$\tabularnewline
\hline
$D$  & $500\,\mathrm{Kbits}$\tabularnewline
\hline
$R_{\mathrm{e}}$  & $6371\,\mathrm{km}$\tabularnewline
\hline
$T_{\mathrm{slot}}$  & $1\,\mathrm{sec}$\tabularnewline
\hline
\end{tabular}
\end{table}

In this section, numerical results and discussion  will be provided to study the performances of the considered LEO STT system. The simulation methodology is the proposed analytical framework while the analytical results of throughput, EE and DOR with RAT and PAT schemes are obtained from Eqs. (22-23), Eqs. (27-28),  Eq. (34) and  Eq. (41), Eq. (42),  Eq. (49), and  the simulation results are obtained  with Monte-Carlo simulation. The simulation parameters
are listed in Table II and each plot is realized over 10000 channel
realizations.  Besides, the analytical method based on the randomly located terrestrial terminals \cite{Doppler_3} is regarded as
the benchmark.

Figs. \ref{through_RAT}, \ref{EE_RAT}, and \ref{DORvsT} describe the throughput, EE, and DOR of the RAT scheme, respectively. In Fig. \ref{through_RAT}, we illustrate the relationship between the throughput of RAT and the orbit height of $\mathrm{S}$ for various transmit powers at the LEO satellite.
 The throughput increases as $P_\mathrm{T}$ increases or $H$ decreases, due to a larger average channel gain between $\mathrm{S}$ and $\mathrm{T}$ and a larger received SNR at $\mathrm{T}$. Moreover, the simulation results fall between the upper and lower bounds. Furthermore, compared to the results with randomly located terrestrial terminals, the throughput obtained from our proposed analytical framework is superior
to the benchmark, which proves
the superiority of the proposed analytical framework.

Fig. \ref{EE_RAT} presents the EE performance of RAT for different orbit heights of the LEO satellite. Similar to the throughput performance in Fig. \ref{through_RAT}, EE decreases when $P_{\mathrm{T}}$ decreases or the orbit height of the LEO satellite increases, beacause low data rate reduces EE.
The transmit power at the LEO satellite exhibits a positive effect on the EE performance. When the transmit power at the LEO satellite is small, the variation of EE with different orbit heights is small, and is large when the transmit power at the LEO satellite increases. This means that, when the height of the LEO satellite is relatively low, the improvement arising from increasing the transmit power at the LEO satellite is more conspicuous, due to a more sensible enhancement on the data rate compared to the average power consumption. It indicates that increasing the transmit power at the LEO satellite is an effective way to enhance EE when the LEO satellite works on a lower orbit height. Furthermore, similar to Fig. \ref{through_RAT}, simulation results fall between the upper and lower bounds and is superior
to the benchmark.

\begin{figure}[!h]
\setlength{\abovecaptionskip}{0pt}
\setlength{\belowcaptionskip}{10pt}
\centering
\includegraphics[scale=0.46]{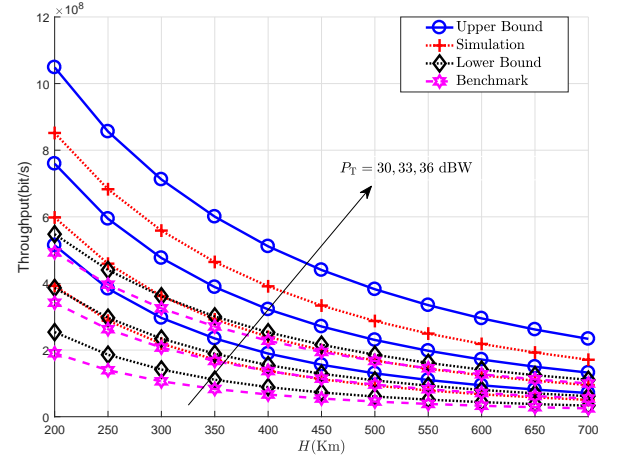}
\caption{Throughput of RAT under different orbit height.}
\label{through_RAT}
\end{figure}
\begin{figure}[!h]
\setlength{\abovecaptionskip}{0pt}
\setlength{\belowcaptionskip}{10pt}
\centering
\includegraphics[scale=0.46]{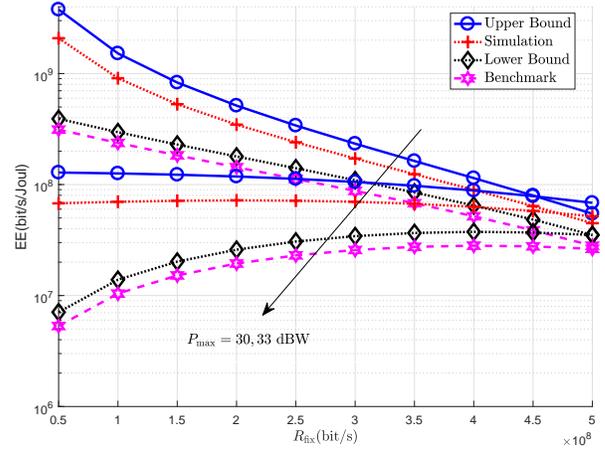}
\caption{EE of RAT under different orbit height.  }
\label{EE_RAT}
\end{figure}
\begin{figure}[!h]
\setlength{\abovecaptionskip}{0pt}
\setlength{\belowcaptionskip}{10pt}
\centering
\includegraphics[scale=0.46]{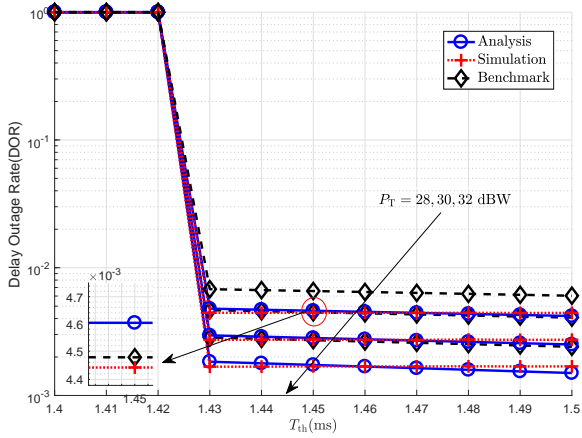}
\caption{Delay outage rate of RAT under different $T_{\mathrm{th}}$.  }
\label{DORvsT}
\end{figure}
\begin{figure}
\setlength{\abovecaptionskip}{0pt}
\setlength{\belowcaptionskip}{10pt}
\begin{centering}
\includegraphics[scale=0.46]{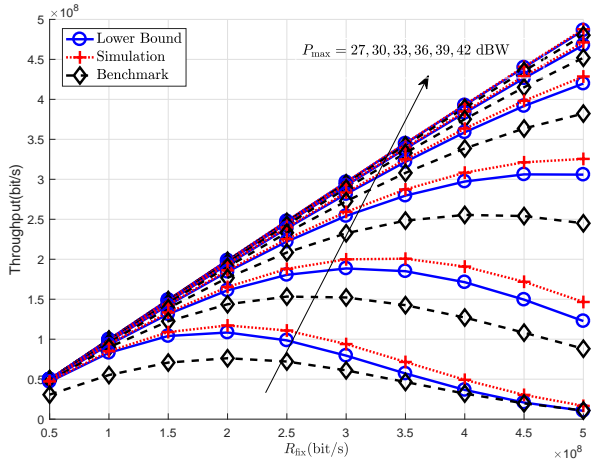}
\par\end{centering}
\caption{Throughput of PAT under different fix data rate.  }
\label{through_PAT}
\end{figure}

In Fig. \ref{DORvsT}, the DOR performance of RAT over the LEO STT link is presented. The analytical results match closely with simulation and DOR is zero when $T_{\mathrm{th}}$ is too small which means the LEO STT link cannot support a very small delay transmission. The DOR declines with the increased transmit power at the LEO satellite or  $T_{\mathrm{th}}$. Increasing the transmit power at the LEO satellite leads to higher received SNR at the terrestrial terminal and larger average throughput which decreases DOR. Moreover, the system is  more tolerant to transmission delay and exhibits better DOR performance with smaller $T_{\mathrm{th}}$.
Furthermore, compared to the results with randomly located terrestrial terminals, the DOR obtained from our proposed analytical framework is lower than that from
benchmark, which proves
the superiority of the proposed analytical framework.

Figs. \ref{through_PAT}, \ref{EE_PAT}, and \ref{DOR_PAT} describe the throughput, EE, and DOR performance of the PAT scheme, respectively. From Fig. \ref{through_PAT}, we could see that the lower bound of the throughput matches closely with the simulation results with $P_\mathrm{T}$ is $36$, $39$  and $42$ dBW, which means the higher the fix data rate is, the larger throughput is. When $P_\mathrm{T}$ is $27$, $30$, and $33$ dBW, there exists a gap between the lower bound and simulation results of the throughput, which is acceptable.  At first, the throughput increases and then decreases with increasing fixed transmission rate. This is because,
when the transmit power at the LEO satellite is relatively small, it  may not afford adequate power to support a higher fixed transmission rate and reduce the throughput with higher fixed transmission rate.  Obviously, the throughput obtained from our proposed analytical framework is superior to the benchmark.
Moreover, compared Figs. \ref{EE_RAT} and \ref{EE_PAT}, one can conclude that the PAT scheme is easier to get better EE performance with similar constraints because the PAT scheme can avoid invalid transmit power at the LEO satellite.

\begin{figure}
\setlength{\abovecaptionskip}{0pt}
\setlength{\belowcaptionskip}{10pt}
\begin{centering}
\includegraphics[scale=0.46]{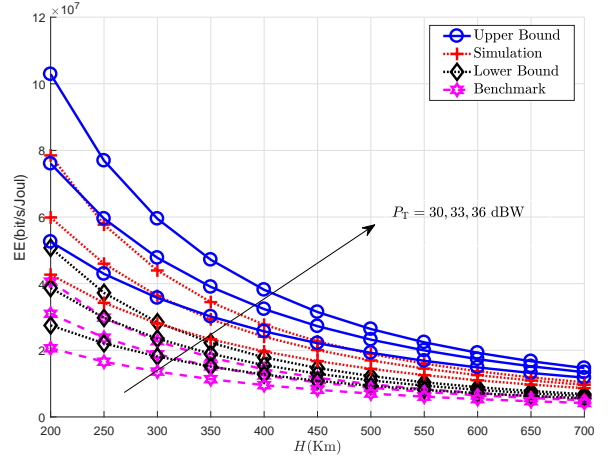}
\par\end{centering}
\caption{EE  of PAT under different fix data rate.}
\label{EE_PAT}
\end{figure}

In Fig. \ref{EE_PAT},  we describe the relationship between the EE  performance of the PAT scheme and fixed data rate in case of various transmit power at the LEO satellite. The transmit power at the LEO satellite shows a positive impact on the EE when the fixed data rate is relatively small and shows a negative impact on the EE when the fixed data rate is relatively large. This is because larger transmit power at the LEO satellite means a better received SNR at the terrestrial terminal but a larger power consumption.
 When the transmit power at the LEO satellite is relatively large and the fixed date rate is relatively small, the average throughput can compensate for the power consumption and exhibit a positive impact on the EE performance.
 Thus, the influence of the fixed data rate reflects the effective utilization of the transmit power at the LEO satellite which decreases EE. When the fixed data rate is too great or too small, the transmit power at the LEO satellite cannot be fully utilized and result in the decline of EE. Similar to Fig. \ref{EE_RAT}, simulation results fall between the upper and lower bounds and  is superior
to the benchmark. Compared Figs. \ref{through_RAT} and \ref{through_PAT}, one can conclude that the RAT scheme is easier to realize high throughput with similar constraints because the RAT scheme can use CSI more effectively.

In Fig. \ref{DOR_PAT}, the DOR performance of PAT over the LEO STT link is presented. DOR equals 1 when  $T_{\mathrm{th}}=0$, which means that the LEO STT link cannot afford zero delay transmission. The larger the maximum transmit power at the LEO satellite is, the smaller the DOR obtained by the PAT scheme is. This is easy to understand that the maximum transmit power at the LEO satellite can lead to a stabilized received SNR with fix data rate and then improve the DOR performance. DOR decreases with increasing $T_{\mathrm{th}}$. Moreover, the upper and lower bounds show first steep and then more stationary slope, this is because that we adopt the channel gain discretization and time discretization.
 DOR performance can give hints to researchers to design systems considering both  reliability and latency, because absolute reliability and ultra-low delay cannot be  satisfied at the same time, but a trade-off can be arrived to get relatively  high reliability and acceptable delay with exiting constrains.
At the same time, simulation results fall between the upper and lower bounds  and is superior
to the benchmark, which is capable of the DOR analysis.  Comparing Figs. \ref{DORvsT} and \ref{DOR_PAT}, we can see that, with relatively large transmit power at the LEO satellite, the RAT scheme can offer a lower DOR with larger $T_{\mathrm{th}}$ and PAT scheme can realize a lower DOR with smaller $T_{\mathrm{th}}$. This is because the CSI is   more effectively used  with larger $T_{\mathrm{th}}$ for the RAT scheme   and higher data rate is more reliably afforded with lower $T_{\mathrm{th}}$ for the  PAT scheme.

\begin{figure}
\setlength{\abovecaptionskip}{0pt}
\setlength{\belowcaptionskip}{10pt}
\begin{centering}
\includegraphics[scale=0.46]{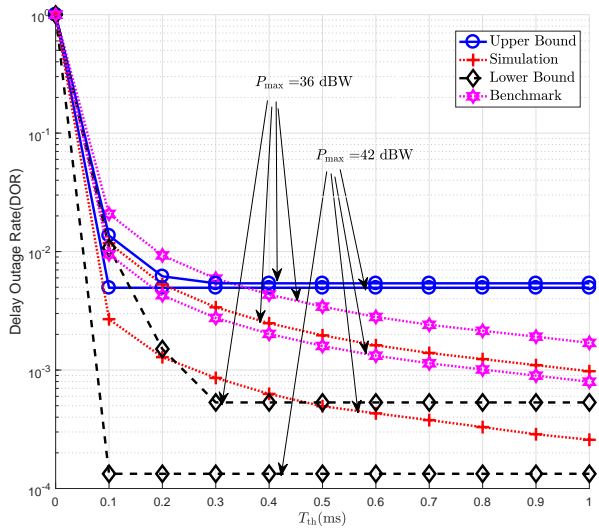}
\par\end{centering}
\caption{Delay outage rate  of PAT under different maximum transmit power at the LEO satellite.}
\label{DOR_PAT}
\end{figure}
\section{Conclusion}
Taking the strong time-varying transmission distance into account, we propose a new analytical framework of using FSMC model and time discretization method, and evaluate the performance of the considered LEO STT system. On one hand, the FSMC model is adopted to analyze the effect of the small-scale fading which is described as the SR channel model over LEO STT links. On the other hand,  the time discretization method is employed to reflect the large-scale fading, because the transmission distance over LEO STT links strongly changes as the LEO satellite flies fast in its orbit. To demonstrate the applications of the proposed framework, we investigate the capacity, energy efficiency, and outage rate performance of the considered LEO STT scenarios with RAT and PAT schemes. Closed-form expressions for throughput, EE, and DOR of the considered LEO STT scenarios with the proposed framework are derived and verified.
Furthermore, we would extend the proposed analytical
framework to scenarios with relays or with outdated CSI in our future work.

\appendices

\section{Distance between $\mathrm{S}$ and $\mathrm{T}$}

\begin{figure}[htb]
\begin{centering}
\includegraphics[width=3.5 in]{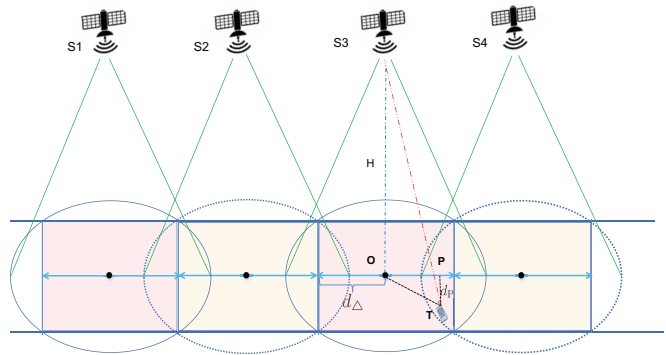}
\par\end{centering}
\caption{Distance between the LEO satellite and the terrestrial terminal at different time}
\label{move2}
\end{figure}
\begin{figure*}[htb]
\centering{}\includegraphics[scale=0.48]{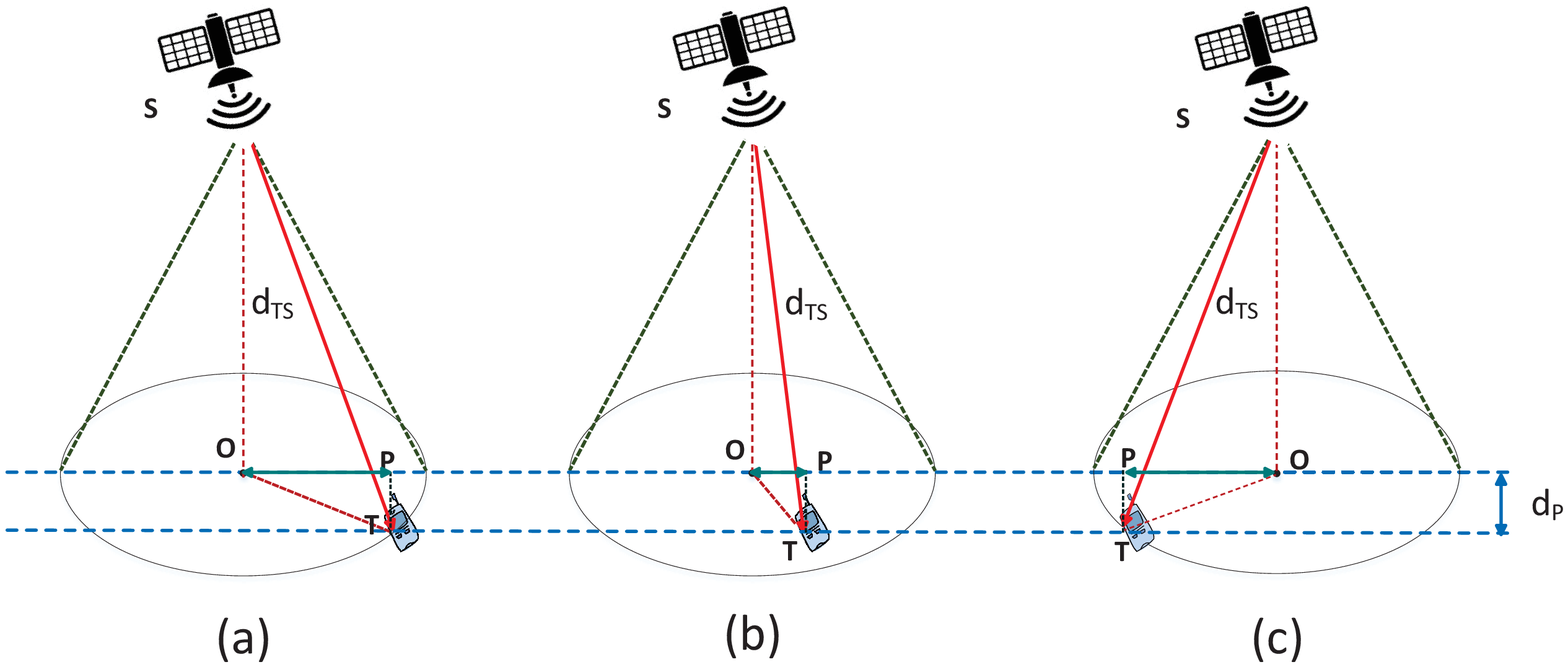}
\caption{Distance between the LEO satellite and the terrestrial terminal at a different time under the service of the LEO satellite. We assume that the LEO satellite moves from left to right. In (a) and (c), the terrestrial terminal is just in the coverage area of the LEO satellite and is on the boundary of the coverage area. The  LEO STT link is just set in (a). As the satellite moves with a certain speed, the LEO STT link is valid like in (b) and is invalid when the satellite depart from the location like in (c).}
\label{move1}
\end{figure*}

As depicted in Figs.  \ref{move2} and \ref{move1}, the distance between $\mathrm{T}$ and $\mathrm{S}$  is
\begin{equation}
d_{\mathrm{TS}}=\sqrt{d_{\mathrm{P}}^{2}+d_{\mathrm{OP}}^{2}+H^{2}}.
\end{equation}

As $\mathrm{S}$   moves with a certain speed $v_{\mathrm{sat}}$, the speed of sub-satellite point $v=v_{\mathrm{sat}}\cdot\frac{R_{\mathrm{e}}}{R_{\mathrm{e}}+H}$,
where $R_{\mathrm{e}}$ is the radius of the earth and $H$ is the LEO satellites' orbital altitude. The  total service duration that the terminal is within
the satellite coverage is denoted as $T_{\mathrm{s}}=\frac{2d_{\triangle}}{v}$. As depicted in Fig. \ref{move2}, $d_{\triangle}$ is fixed and is defined by the annular covering belt, where the continuous coverage area is determined by the coverage of all the LEO satellites circling the earth with the same orbit height. For a given LEO satellite belt,  the service duration is fixed if $\mathrm{T}$  is in the coverage area, regardless of the distance between $\mathrm{T}$ and $\mathrm{S}$.

Next we analyze the time-varying distance between $\mathrm{T}$ and $\mathrm{S}$. As sub-satellite point moves with a certain speed $v$, the distance between $\mathrm{P}$ and $\mathrm{O}$ is
\begin{equation}
d_{\mathrm{OP}}\left(t\right)=\left|d_{\triangle}-vt\right|,
\end{equation}
where $0\leq t\leq T_{\mathrm{s}}$ and $0\leq d_{\mathrm{OP}}\leq d_{\triangle}$.

Then, the distance between $\mathrm{T}$ and $\mathrm{S}$  is obtained as
\begin{equation}
d_{\mathrm{TS}}\left(t\right)=\sqrt{\left|d_{\triangle}-vt\right|^{2}+d_{\mathrm{P}}^{2}+H^{2}}.
\end{equation}
When $\mathrm{T}$  is just within the coverage area of the $\mathrm{S}$  and is on the boundary of the coverage area, $d_{\mathrm{TS}}\left(0\right)=d_{\mathrm{TS}}\left(T_{\mathrm{s}}\right)=\sqrt{H^{2}+R^{2}}$, where
 $R$ is the coverage area radius of LEO satellite.

When the satellite is right above the terminal and $\mathrm{T}$  is on the track of sub-satellite point, $d_{\mathrm{TS}}\left(\frac{T_{\mathrm{s}}}{2}\right)=H$.

Considering all the terminals within the continuous service of  $\mathrm{S}$, the range of $d_{\mathrm{TS}}$ is

\begin{equation}\label{dTS_all}
H\leq d_{\mathrm{TS}}\leq\sqrt{H^{2}+R^{2}}.
\end{equation}
The maximum of the distance between $\mathrm{T}$  and $\mathrm{S}$ is defined as  $d_{\mathrm{TS-max}}$ and can be given as

\begin{equation}
d_{\mathrm{TS-max}}=\sqrt{H^{2}+R^{2}}.
\end{equation}

When $\mathrm{T}$ is fixed and the location of $\mathrm{T}$ is described as $d_{\mathrm{P}}$, the distance between  $\mathrm{S}$  and  $\mathrm{T}$ is
\begin{equation}\label{dTS_fixed}
\sqrt{d_{\mathrm{P}}^{2}+H^{2}}\leq d_{\mathrm{TS}}\leq\sqrt{{d_{\triangle}}^{2}+d_{\mathrm{P}}^{2}+H^{2}},
\end{equation}
which is smallest when  $\mathrm{S}$  moves to the projection point of $\mathrm{T}$  on the track of sub-satellite point as depicted in Fig. \ref{move2}-a and \ref{move2}-c and is largest when $\mathrm{T}$ is on the edge of the coverage area of $\mathrm{S}$. To summarize,  the range of $d_{\mathrm{TS}}$ is expressed in \eqref{dTS_all} considering all the  terminals randomly located in the service area of S, and the range of $d_{\mathrm{TS}}$ is presented in \eqref{dTS_fixed} when $\mathrm{T}$ is fixed and the location of $\mathrm{T}$ is described as $d_{\mathrm{P}}$.
\section{Proof of Theorem \ref{DOR1_theorem}}
 Based on \eqref{dor_def}, the average DOR with RAT scheme, denoted by $\mathrm{DOR}_{\mathrm{R}}$, can be given as \eqref{DOR1_integral}, shown on the top of next page.
\begin{figure*}
 \centering
\begin{align}\label{DOR1_integral}
&\mathrm{DOR_{\mathrm{\mathrm{R}}}}=\frac{1}{T_{\mathrm{s}}}\int_{0}^{T_{\mathrm{s}}}\mathrm{DOR_{\mathrm{R}}}\left(t\right)dt\notag\\
&= \frac{1}{T_{\mathrm{s}}}\int_{0}^{T_{\mathrm{s}}}  \left\{ 1-\sum_{m_t=1}^{N}\frac{\pi_{1,m_t}}{N}\left[1-\exp\left(-\frac{T_{\mathrm{th}}-\frac{D}{R_{2,m_t}^{\mathrm{RAT}}}}{\lambda}\right)\right]\mathcal{U}\left(T_{\mathrm{th}}-\frac{D}{R_{2,m_t}^{\mathrm{RAT}}}\right)-
\frac{1}{N}\sum_{m_t=1}^{N}\sum_{k=2}^{K}\pi_{k,m_t}\mathcal{U}\left(T_{\mathrm{th}}-\frac{D}{R_{k,m_t}^{\mathrm{RAT}}}\right) \right\}dt
\end{align}
\rule{18cm}{0.01cm}
\end{figure*}

The result of the first integral term in  \eqref{DOR1_integral} is 1.
The arriving time of the data packet is uniformly distributed, and
all of the waiting period is subjected to the same distribution, so
$m_t=\left\lceil \frac{t+T_{\mathrm{W}}}{T_{\mathrm{slot}}}\right\rceil \,\mathrm{MOD}\,N$
is uniformly distributed in the $N$ time slot. In other words, we have
\begin{equation}
\Pr\left\{ m_t=n\right\} =\frac{1}{N},\:n=1,\cdots,N.
\end{equation}
%
%

By using \eqref{m_equation}, the second integral term in  \eqref{DOR1_integral} is
\begin{align}
\mathcal{I}_{1}=&\frac{1}{N}\sum_{m_t=1}^{N}\pi_{1,m_t}\left[1-\exp\left(-\frac{T_{\mathrm{th}}-\frac{D}{R_{2,m_t}^{\mathrm{RAT}}}}{\lambda}\right)\right]
\notag\\
&\times
\mathcal{U}\left(T_{\mathrm{th}}-\frac{D}{R_{2,m_t}^{\mathrm{RAT}}}\right),
\end{align}
in which
\begin{align}
R_{2,m_t}^{\mathrm{RAT}}&=B\log_{2}\left(1+\gamma_{2,m_t}\right) \notag\\
&=B\log_{2}\left(1+\frac{P_{\mathrm{T}}}{\sigma^{2}}\frac{\mu_{1}^{2}}{d_{\mathrm{TS}}^{\rho}[m_t]}\right) \notag\\
&=B\log_{2}\left(1+\frac{\gamma_{\mathrm{min}}d_{\mathrm{max}}^{\rho}}{d_{\mathrm{TS}}^{\rho}[m_t]}\right).
\end{align}
Thus, the second integral term in \eqref{DOR1_integral} is
\begin{align}\label{I1}
&\mathcal{I}_{1}=\frac{1}{N}\sum_{m_t=1}^{N}\pi_{1,m_t}\mathcal{U}\left(T_{\mathrm{th}}-\frac{D}{B\log_{2}\left(1+\frac{\gamma_{\mathrm{min}}d_{\mathrm{max}}^{\rho}}{d_{\mathrm{TS}}^{\rho}[m_t]}\right)}\right)
\notag\\
&\times \left[1-\exp\left(-\frac{1}{\lambda}\left(T_{\mathrm{th}}-\frac{D}{B\log_{2}\left(1+\frac{\gamma_{\mathrm{min}}d_{\mathrm{max}}^{\rho}}{d_{\mathrm{TS}}^{\rho}[m_t]}\right)}\right)\right)\right]
.\end{align}

%
%
%


We partition the service time of each satellite with equal intervals, and the satellite moves with steady speed. Thus, the third integral term in \eqref{DOR1_integral} is
\begin{align}\label{third}
&\frac{1}{T}\int_{0}^{T}\frac{1}{N}\sum_{m_t=1}^{N}\sum_{k=2}^{K}\pi_{k,m_t}\mathcal{U}\left(T_{\mathrm{th}}-\frac{D}{R_{k,m_t }^{\mathrm{RAT}}}\right)dt \notag\\
&=\frac{1}{N}\sum_{n=1}^{N}\sum_{k=2}^{M}\pi_{k,n}\mathcal{U}\left(T_{\mathrm{th}}-\frac{D}{B\log_{2}\left(1+\frac{P_{\mathrm{T}}}{\sigma^{2}}\frac{\mu_{k-1}^{2}}{d_{\mathrm{TS}}^{\rho}[n]}\right)}\right).
\end{align}

Combining \eqref{DOR1_integral}, \eqref{I1} and \eqref{third}, the proof of \emph{Theorem \ref{DOR1_theorem}} is completed.

\begin{IEEEbiography}
[{\includegraphics[width=1in,height=1.25in,clip,keepaspectratio]{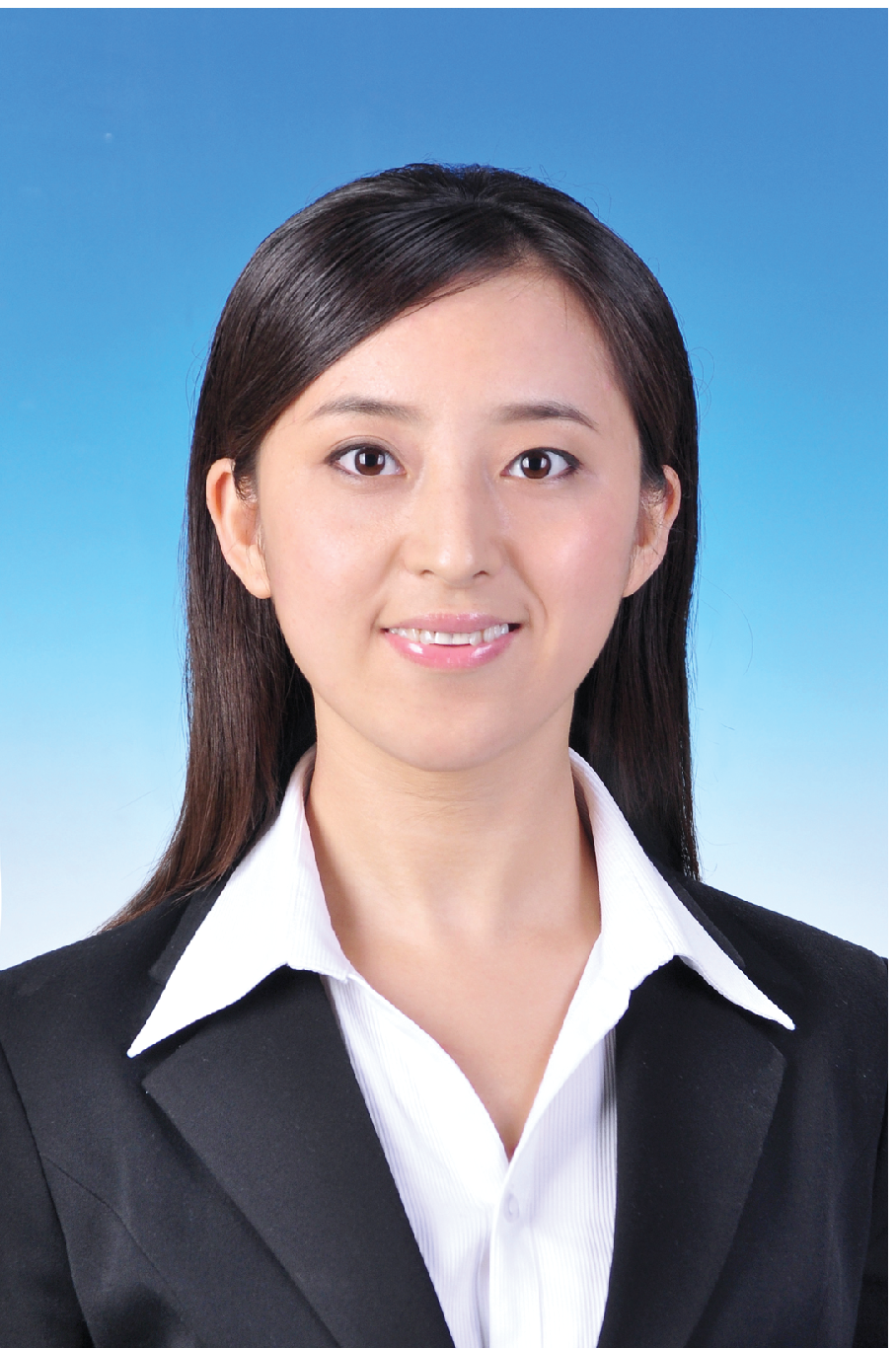}}]{Yuanyuan Ma}
 received the B.Sc  degree in communication engineering from Henan Normal University, China, in 2011, and the M.S. degree in information and communication engineering from Beijing Institute of Technology, China, in 2014. She is currently pursuing the Ph.D. degree in information and communication engineering from Beijing University of Posts and Telecommunications, China. Her current research interests include satellite communication, cognitive radio networks and performance analysis.
\end{IEEEbiography}

\begin{IEEEbiography}[{\includegraphics[width=1in,height=1.25in,clip,keepaspectratio]{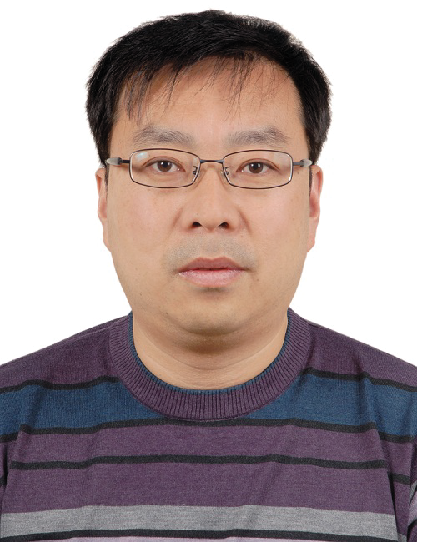}}]{Tiejun Lv}
(M'08-SM'12) received the M.S. and Ph.D. degrees in electronic engineering from the University of Electronic Science and Technology of China (UESTC), Chengdu, China, in 1997 and 2000, respectively. From January 2001 to January 2003, he was a Postdoctoral Fellow with Tsinghua University, Beijing, China. In 2005, he was promoted to a Full Professor with the School of Information and Communication Engineering, Beijing University of Posts and Telecommunications (BUPT). From September 2008 to March 2009, he was a Visiting Professor with the Department of Electrical Engineering, Stanford University, Stanford, CA, USA. He is the author of three books, more than 100 published IEEE journal papers and 200 conference papers on the physical layer of wireless mobile communications. His current research interests include signal processing, communications theory and networking. He was the recipient of the Program for New Century Excellent Talents in University Award from the Ministry of Education, China, in 2006. He received the Nature Science Award in the Ministry of Education of China for the hierarchical cooperative communication theory and technologies in 2015.
\end{IEEEbiography}

\begin{IEEEbiography}
[{\includegraphics[width=1in,height=1.25in,clip,keepaspectratio]{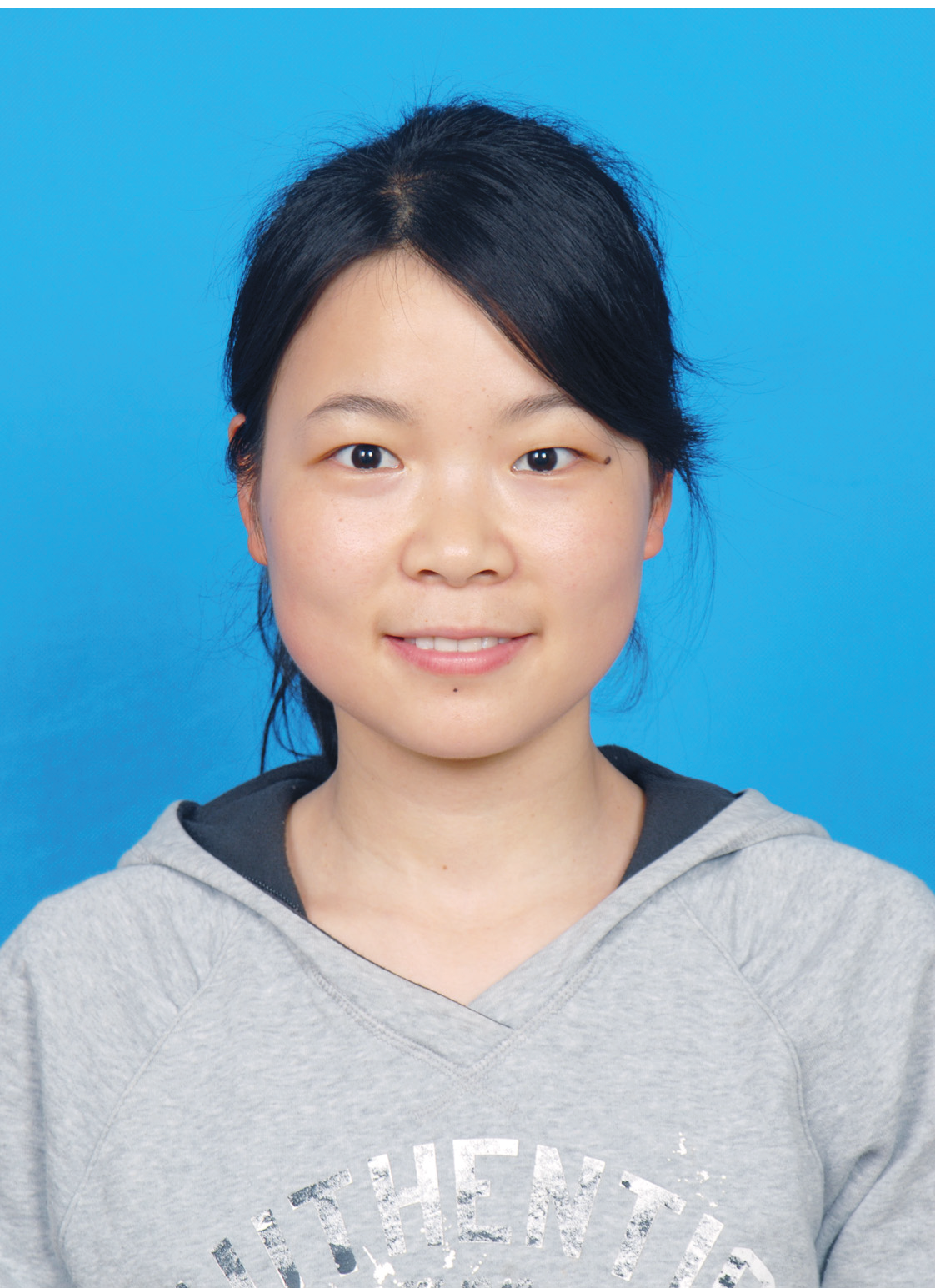}}]{Tingting Li}
received the B.Sc. degree in mathematics and applied mathematics and the Ph.D. degree in computational mathematics
from Chongqing University, Chongqing, China, in 2006 and 2012, respectively. In July 2012, she joined the School of Mathematics and Statistics,
Southwest University, Chongqing, China, where she is currently an Associate Professor. Her research interests include statistics and its applications.
\end{IEEEbiography}

\begin{IEEEbiography}
[{\includegraphics[width=1in,height=1.25in,clip,keepaspectratio]{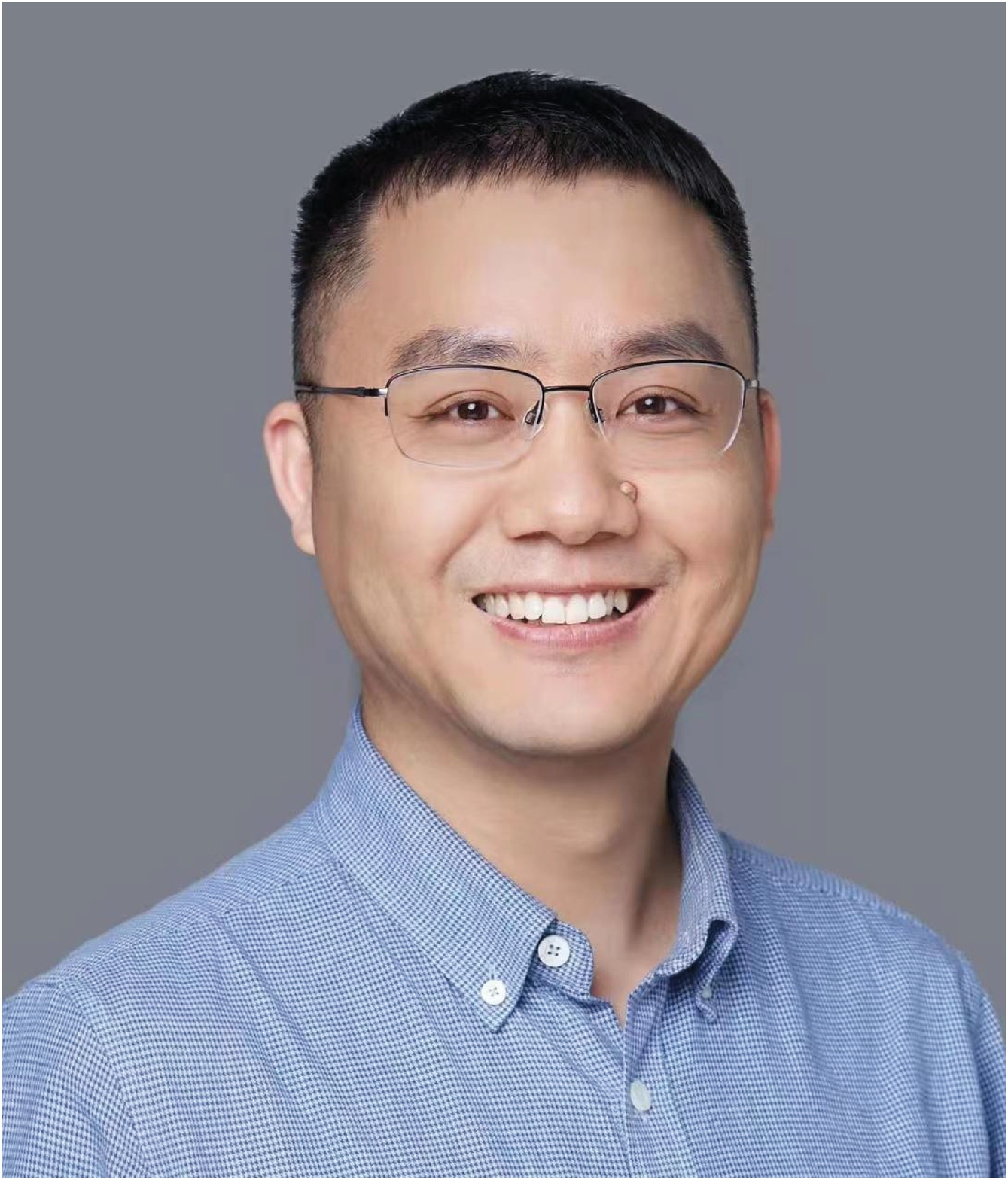}}]{Gaofeng Pan}
(Senior Member, IEEE) received
his B.Sc in Communication Engineering from
Zhengzhou University, Zhengzhou, China, in 2005,
and the Ph.D. degree in Communication and Information Systems from Southwest Jiaotong University,
Chengdu, China, in 2011. He is currently with the
School of Cyberspace Science and Technology, Beijing Institute of Technology, China, as a Professor.
His research interest spans special topics in communications theory, signal processing, and protocol
design.
\end{IEEEbiography}

\begin{IEEEbiography}
[{\includegraphics[width=1in,height=1.25in,clip,keepaspectratio]{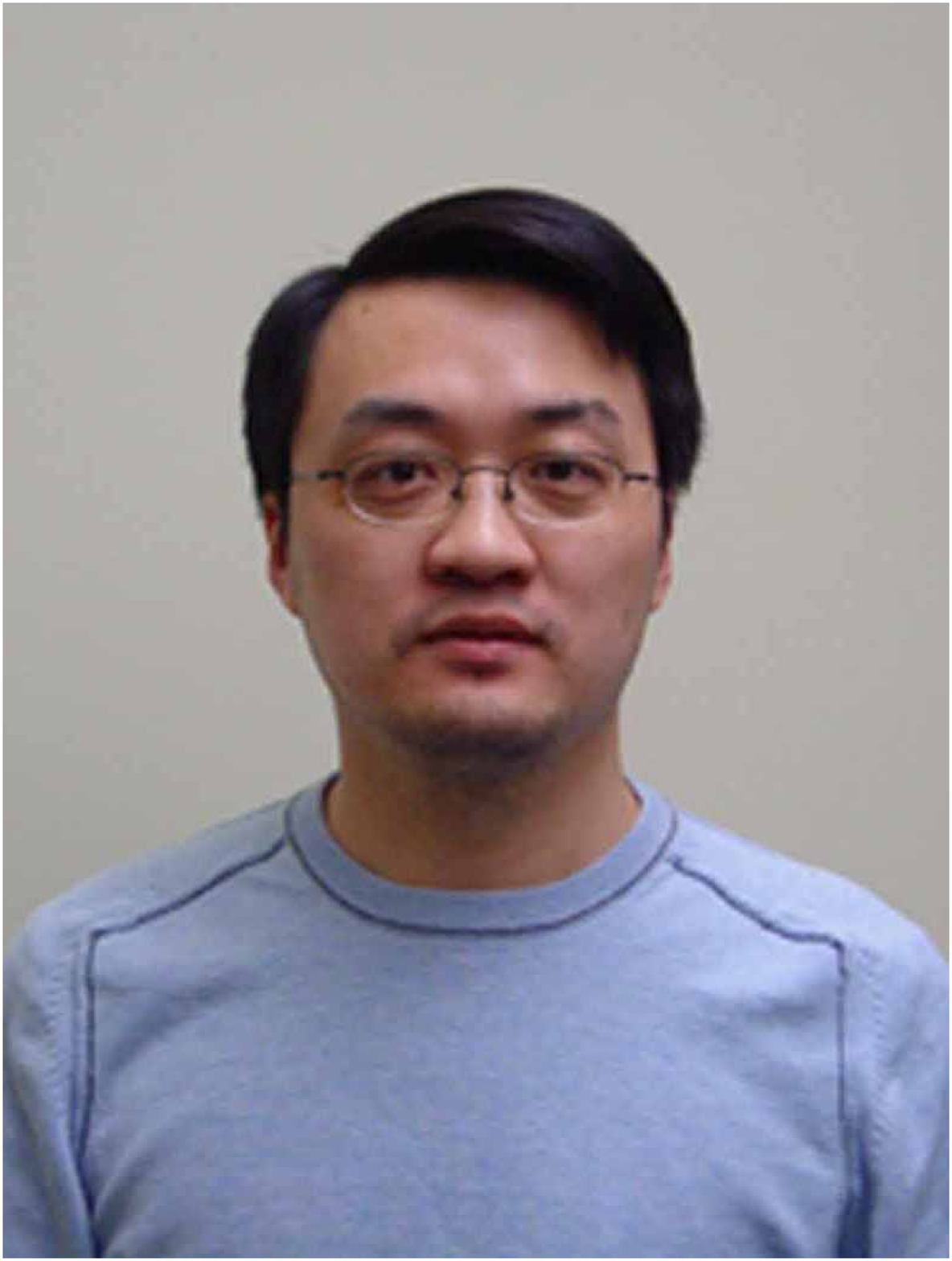}}]{Yunfei Chen}
(S’02-M’06-SM’10) received his B.E. and M.E. degrees in electronics engineering from Shanghai Jiaotong University, Shanghai, P.R.China, in 1998 and 2001, respectively. He received his Ph.D. degree from the University of Alberta in 2006. He is currently working as an Associate Professor at the University of Warwick, U.K. His research interests include wireless communications, cognitive radios, wireless relaying and energy harvesting.
\end{IEEEbiography}

\begin{IEEEbiography}[{\includegraphics[width=1in,height=1.25in,clip,keepaspectratio]{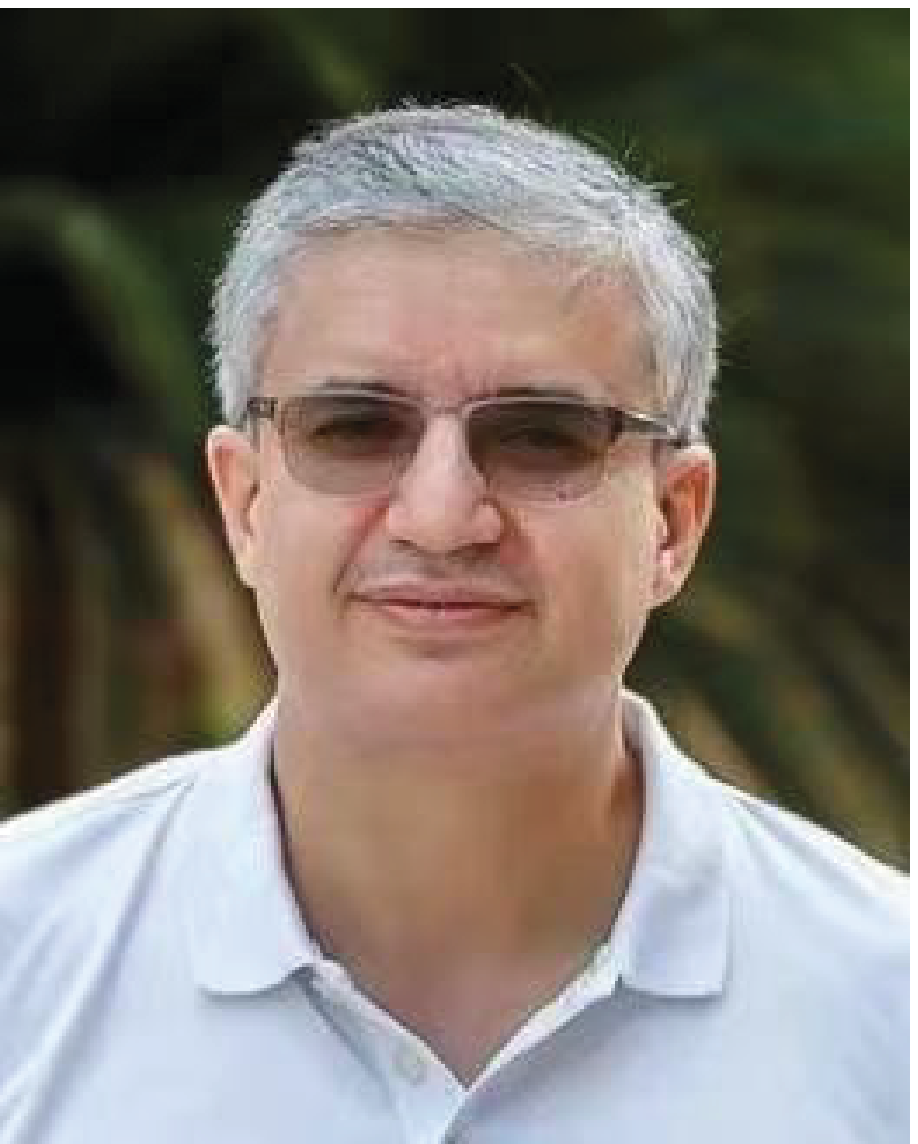}}]{Mohamed-Slim Alouini }(Fellow, IEEE) was born in Tunis, Tunisia. He received the Ph.D. degree in Electrical Engineering from the California Institute of Technology (Caltech), Pasadena, CA, USA, in 1998. He served as a faculty member in the University of Minnesota, Minneapolis, MN, USA, then in the Texas A\&M University at Qatar, Education City, Doha, Qatar before joining King Abdullah University of Science and Technology (KAUST), Thuwal, Makkah Province, Saudi Arabia as a Professor of Electrical Engineering in 2009. His current research interests include modeling, design, and performance analysis of wireless communication systems.
\end{IEEEbiography}


\begin{thebibliography}{10}

\bibitem{Ye1}
A.~Guidotti, A.~Vanelli-Coralli, M.~Conti, S.~Andrenacci, S.~Chatzinotas,
  N.~Maturo, B.~G. Evans, A.~B. Awoseyila, A.~Ugolini, T.~Foggi, L.~Gaudio,
  N.~Alagha, and S.~Cioni, ``Architectures and key technical challenges for
  {5G} systems incorporating satellites.'' \emph{IEEE Trans. Veh. Technol},
  vol.~68, no.~3, pp. 2624--2639, 2019.

\bibitem{6G1}
C.~D. Alwis, A.~Kalla, Q.-V. Pham, P.~Kumar, K.~Dev, W.-J. Hwang, and
  M.~Liyanage, ``Survey on 6{G} frontiers: Trends, applications, requirements,
  technologies and future research,'' \emph{IEEE Open J. Commun. Soc.}, vol.~2,
  pp. 836--886, 2021.

\bibitem{6G2}
H.~Tataria, M.~Shafi, A.~F. Molisch, M.~Dohler, H.~Sj\"{o}land, and
  F.~Tufvesson, ``{6G} wireless systems: Vision, requirements, challenges,
  insights, and opportunities,'' \emph{Proc. IEEE}, vol. 109, no.~7, pp.
  1166--1199, 2021.

\bibitem{Ye4}
E.~Yaacoub and M.-S. Alouini, ``A key {6G} challenge and opportunity -
  connecting the base of the pyramid: A survey on rural connectivity.''
  \emph{Proc. IEEE}, vol. 108, no.~4, pp. 533--582, 2020.

\bibitem{9760009}
G.~Pan, J.~Ye, J.~An, and S.~Alouini, ``Latency versus reliability in {LEO}
  mega-constellations: Terrestrial, aerial, or space relay,'' \emph{IEEE Trans.
  Mobile Comput.}, pp. 1--1, 2022, doi={10.1109/TMC.2022.3168081}.

\bibitem{multisat_multibeam}
J.~Romero{-}Garc{\'{\i}}a and R.~D. Gaudenzi, ``On antenna design and capacity
  analysis for the forward link of a multibeam power controlled satellite
  {CDMA} network,'' \emph{{IEEE} J. Sel. Areas Commun.}, vol.~18, no.~7, pp.
  1230--1244, 2000.

\bibitem{CDMA_cellular}
S.~Vassaki, A.~D. Panagopoulos, and P.~Constantinou, ``Effective capacity and
  optimal power allocation for mobile satellite systems and services.''
  \emph{IEEE Commun. Lett.}, vol.~16, no.~1, pp. 60--63, 2012.

\bibitem{Lin_BF_rate}
Z.~Lin, M.~Lin, W.-P. Zhu, J.-B. Wang, and J.~Cheng, ``Robust secure
  beamforming for wireless powered cognitive satellite-terrestrial networks,''
  \emph{IEEE Trans. Cogn. Commun. Netw.}, vol.~7, no.~2, pp. 567 --580, 2021.

\bibitem{An_add_re3}
X.~Zhang, D.~Guo, K.~An, Z.~Chen, B.~Zhao, Y.~Ni, and B.~Zhang, ``Performance
  analysis of {NOMA}-based cooperative spectrum sharing in hybrid
  satellite-terrestrial networks,'' \emph{IEEE Access}, vol.~7, pp.
  172\,321--172\,329, 2019.

\bibitem{EE_SIOT}
J.~Jiao, Z.~Ni, S.~Wu, Y.~Wang, and Q.~Zhang, ``Energy efficient network coding
  {HARQ} transmission scheme for {S-IoT},'' \emph{IEEE Trans. Green Commun. and
  Netw.}, vol.~5, no.~1, pp. 308--321, 2021.

\bibitem{EE_ref_Z_Lin2}
Z.~Lin, M.~Lin, B.~Champagne, W.-P. Zhu, and N.~Al-Dhahir, ``Secure and energy
  efficient transmission for {RSMA}-based cognitive satellite-terrestrial
  networks,'' \emph{IEEE Wireless Commun. Lett.}, vol.~10, no.~2, pp. 251--255,
  2021.

\bibitem{Guo_SOP}
K.~Guo, K.~An, B.~Zhang, Y.~Huang, X.~Tang, G.~Zheng, and T.~A. Tsiftsis,
  ``Physical layer security for multiuser satellite communication systems with
  threshold-based scheduling scheme.'' \emph{IEEE Trans. Veh. Technol},
  vol.~69, no.~5, pp. 5129--5141, 2020.

\bibitem{DVB-S2}
S.~Enserink, A.~D. Panagopoulos, and M.~P. Fitz, ``On the calculation of
  constrained capacity and outage probability of broadband satellite
  communication links,'' \emph{IEEE Wireless Commun. Lett.}, vol.~3, no.~5, pp.
  453--456, 2014.

\bibitem{Pan_3}
G.~Pan, J.~Ye, Y.~Tian, and M.-S. Alouini, ``On {HARQ} schemes in
  satellite-terrestrial transmissions,'' \emph{IEEE Trans. Wireless Commun.},
  vol.~19, no.~12, pp. 7998--8010, 2020.

\bibitem{Tian_Yu}
Y.~Tian, G.~Pan, M.~A. Kishk, and M.-S. Alouini, ``Stochastic analysis of
  cooperative satellite-{UAV} communications,'' \emph{IEEE Trans. Wireless
  Commun.}, pp. 1--1, 2021.

\bibitem{Add_re1_4}
X.~Zhang, K.~An, B.~Zhang, Z.~Chen, Y.~Yan, and D.~Guo, ``Vickrey auction-based
  secondary relay selection in cognitive hybrid satellite-terrestrial overlay
  networks with non-orthogonal multiple access,'' \emph{IEEE Wireless Commun.
  Lett.}, vol.~9, no.~5, pp. 628--632, 2020.

\bibitem{Add_re1_6}
X.~Zhang, B.~Zhang, K.~An, B.~Zhao, Y.~Jia, Z.~Chen, and D.~Guo, ``On the
  performance of hybrid satellite-terrestrial content delivery networks with
  non-orthogonal multiple access,'' \emph{IEEE Wireless Commun. Lett.},
  vol.~10, no.~3, pp. 454--458, 2021.

\bibitem{Pan_2}
G.~Pan, J.~Ye, Y.~Zhang, and M.-S. Alouini, ``Performance analysis and
  optimization of cooperative satellite-aerial-terrestrial systems,''
  \emph{IEEE Trans. Wireless Commun.}, vol.~19, no.~10, pp. 6693--6707, 2020.

\bibitem{LEO_6}
Y.~Ruan, Y.~Li, C.-X. Wang, R.~Zhang, and H.~Zhang, ``Performance evaluation
  for underlay cognitive satellite-terrestrial cooperative networks.''
  \emph{Sci. China Inf. Sci.}, vol.~61, no.~10, pp. 213--223, 2018.

\bibitem{LEO_8}
Y.~Su, Y.~Liu, Y.~Zhou, J.~Yuan, H.~Cao, and J.~Shi, ``Broadband {LEO}
  satellite communications: Architectures and key technologies.'' \emph{IEEE
  Wirel. Commun.}, vol.~26, no.~2, pp. 55--61, 2019.

\bibitem{LEO}
R.~Deng, B.~Di, and L.~Song, ``Ultra-dense {LEO} satellite based formation
  flying,'' \emph{IEEE Trans. Commun.}, vol.~69, no.~5, pp. 3091--3105, 2021.

\bibitem{Add_re1_1}
R.~Deng, B.~Di, S.~Chen, S.~Sun, and L.~Song, ``Ultra-dense {LEO} satellite
  offloading for terrestrial networks: How much to pay the satellite
  operator?'' \emph{IEEE Trans. Wireless Commun.}, vol.~19, no.~10, pp.
  6240--6254, 2020.

\bibitem{Add_re1_2}
B.~Di, H.~Zhang, L.~Song, Y.~Li, and G.~Y. Li, ``Ultra-dense {LEO}: Integrating
  terrestrial-satellite networks into {5G} and beyond for data offloading,''
  \emph{IEEE Trans. Wireless Commun.}, vol.~18, no.~1, pp. 47--62, 2019.

\bibitem{Rice}
A.~Abdi, W.~Lau, M.-S. Alouini, and M.~Kaveh, ``A new simple model for land
  mobile satellite channels: first- and second-order statistics,'' \emph{IEEE
  Trans. Wireless Commun.}, vol.~2, no.~3, pp. 519--528, 2003.

\bibitem{satellite-link}
{Chun Loo}, ``A statistical model for a land mobile satellite link,''
  \emph{IEEE Trans. Veh. Technol}, vol.~34, no.~3, pp. 122--127, 1985.

\bibitem{Add_re1_3}
X.~Zhang, B.~Zhang, K.~An, G.~Zheng, S.~Chatzinotas, and D.~Guo, ``Stochastic
  geometry-based analysis of cache-enabled hybrid satellite-aerial-terrestrial
  networks with non-orthogonal multiple access,'' \emph{IEEE Trans. Wireless
  Commun.}, vol.~21, no.~2, pp. 1272--1287, 2022.

\bibitem{Add_re1_5}
X.~Zhang, D.~Guo, K.~An, G.~Zheng, S.~Chatzinotas, and B.~Zhang,
  ``Auction-based multichannel cooperative spectrum sharing in hybrid
  satellite-terrestrial {IoT} networks,'' \emph{IEEE Internet Things J.},
  vol.~8, no.~8, pp. 7009--7023, 2021.

\bibitem{Yejia}
J.~Ye, G.~Pan, and M.-S. Alouini, ``Earth rotation-aware non-stationary
  satellite communication systems: Modeling and analysis,'' \emph{IEEE Trans.
  Wireless Commun.}, vol.~20, no.~9, pp. 5942--5956, 2021.

\bibitem{DBLP}
F.~Bastia, C.~Bersani, E.~A. Candreva, S.~Cioni, G.~E. Corazza, M.~Neri,
  C.~Palestini, M.~Papaleo, S.~Rosati, and A.~Vanelli{-}Coralli, ``{LTE}
  adaptation for mobile broadband satellite networks,'' \emph{{EURASIP} J.
  Wirel. Commun. Netw.}, vol. 2009, pp. 1--13, 2009.

\bibitem{ACM_Sate_EHF}
K.-M. Ekerete, A.~Awoseyila, and B.~Evans, ``Robust adaptive margin for {ACM}
  in satellite links at {EHF} bands,'' \emph{IEEE Commun. Lett.}, vol.~24,
  no.~1, pp. 169--172, 2020.

\bibitem{EE_adaptive}
Y.~Ruan, Y.~Li, C.-X. Wang, and R.~Zhang, ``Energy efficient adaptive
  transmissions in integrated satellite-terrestrial networks with {SER}
  constraints,'' \emph{IEEE Trans. Wireless Commun.}, vol.~17, no.~1, pp.
  210--222, 2018.

\bibitem{ACM_multibeam}
M.~A. Vazquez~Castro and G.~S. Granados, ``Cross-layer packet scheduler design
  of a multibeam broadband satellite system with adaptive coding and
  modulation,'' \emph{IEEE Trans. Wireless Commun.}, vol.~6, no.~1, pp.
  248--258, 2007.

\bibitem{Power_Control_Cognitive_Satellite}
Z.~Chen, B.~Zhao, K.~An, G.~Ding, X.~Zhang, J.~Xu, and D.~Guo, ``Correlated
  equilibrium based distributed power control in cognitive
  satellite-terrestrial networks,'' \emph{IEEE Commun. Lett.}, vol.~25, no.~3,
  pp. 945--949, 2021.

\bibitem{adaptive}
K.~An and T.~Liang, ``Hybrid satellite-terrestrial relay networks with adaptive
  transmission,'' \emph{IEEE Trans. Veh. Technol}, vol.~68, no.~12, pp.
  12\,448--12\,452, 2019.

\bibitem{wc8869710}
F.~Li, K.-Y. Lam, H.-H. Chen, and N.~Zhao, ``Spectral efficiency enhancement in
  satellite mobile communications: A game-theoretical approach,'' \emph{IEEE
  Wireless Commun.}, vol.~27, no.~1, pp. 200--205, 2020.

\bibitem{Interfe_Miti}
L.~Yin, R.~Yang, Y.~Yang, L.~Deng, and S.~Li, ``Beam pointing optimization
  based downlink interference mitigation technique between {NGSO} satellite
  systems,'' \emph{IEEE Wireless Commun. Lett.}, vol.~10, no.~11, pp.
  2388--2392, 2021.

\bibitem{Gradshteyn}
I.~S. Gradshteyn and I.~M. Ryzhik, \emph{Table of integrals, series, and
  products}, 7th~ed.\hskip 1em plus 0.5em minus 0.4em\relax Elsevier/Academic
  Press, Amsterdam, 2007.

\bibitem{finite-state-Markov}
P.~{Sadeghi} and P.~{Rapajic}, ``Capacity analysis for finite-state markov
  mapping of flat-fading channels,'' \emph{IEEE Trans. Commun.}, vol.~53,
  no.~5, pp. 833--840, 2005.

\bibitem{Goldsmith2005Wireless}
A.~Goldsmith, \emph{Wireless Communications, First Edition}, 2005.

\bibitem{Queueing}
L.~M. Surhone, M.~T. Tennoe, and S.~F. Henssonow, \emph{Queueing Theory}.\hskip
  1em plus 0.5em minus 0.4em\relax Betascript Publishing, 2013.

\bibitem{kappa_value}
B.~{Vucetic} and J.~{Du}, ``Channel modeling and simulation in satellite mobile
  communication systems,'' \emph{IEEE J. Sel. Areas Commun.}, vol.~10, no.~8,
  pp. 1209--1218, 1992.

\bibitem{AOA}
A.~{Abdi}, J.~A. {Barger}, and M.~{Kaveh}, ``A parametric model for the
  distribution of the angle of arrival and the associated correlation function
  and power spectrum at the mobile station,'' \emph{IEEE Trans. Veh. Technol},
  vol.~51, no.~3, pp. 425--434, 2002.

\bibitem{Doppler_simu_estimate}
W.~Wang, Y.~Tong, L.~Li, A.-A. Lu, L.~You, and X.~Gao, ``Near optimal timing
  and frequency offset estimation for 5g integrated leo satellite communication
  system,'' \emph{IEEE Access}, vol.~7, pp. 113\,298--113\,310, 2019.

\bibitem{sat_perfect}
G.~{Maral} and M.~{Bousquet}, \emph{{Satellite Communications Systems. Systems,
  Techniques and Technology}}.\hskip 1em plus 0.5em minus 0.4em\relax Wiley,
  New York, 1993.

\bibitem{outdated_CCI_OutCSI}
V.~Bankey and P.~K. Upadhyay, ``Ergodic capacity of multiuser hybrid
  satellite-terrestrial fixed-gain af relay networks with {CCI} and outdated
  {CSI},'' \emph{IEEE Trans. Veh. Technol}, vol.~67, no.~5, pp. 4666--4671,
  2018.

\bibitem{imp_csi1}
X.~Guo, D.~Yang, Z.~Luo, H.~Wang, and J.~Kuang, ``Robust {THP} design for
  energy efficiency of multibeam satellite systems with imperfect {CSI},''
  \emph{IEEE Commun. Lett.}, vol.~24, no.~2, pp. 428--432, 2020.

\bibitem{imp_csi2}
S.~Shi, K.~An, G.~Li, Z.~Li, H.~Zhu, and G.~Zheng, ``Optimal power control in
  cognitive satellite terrestrial networks with imperfect channel state
  information,'' \emph{IEEE Trans. Wireless Commun.}, vol.~7, no.~1, pp.
  34--37, 2018.

\bibitem{Doppler_3}
D.-H. Na, K.-H. Park, Y.-C. Ko, and M.-S. Alouini, ``Performance analysis of
  satellite communication systems with randomly located ground users,''
  \emph{IEEE Trans. Wireless Commun.}, vol.~21, no.~1, pp. 621--634, 2022.

\end{thebibliography}
\end{document}